\documentclass[journal]{IEEEtran}

\usepackage{algorithm}  
\usepackage{algorithmic}
\usepackage{amsfonts}
\usepackage{float}
\usepackage{subfig}
\usepackage{tabularx}
\usepackage{multirow} 
\usepackage{color}
\usepackage{mathrsfs}
\usepackage{bm}
\usepackage{cite}
\usepackage{easyReview}
\usepackage[cmintegrals]{newtxmath}

\usepackage{caption}
\usepackage{amsmath,amsthm}

\theoremstyle{plain}

\newtheorem{theorem}{Theorem}
\newtheorem{corollary}{Corollary}

\newtheorem{proposition}{Proposition}

\renewenvironment{proof}{{\noindent\it Proof:}\quad}{\hfill $\square$\par}

\usepackage{stfloats}

\begin{document}
		%
		\title{RIS with Insufficient Phase Shifting \\Capability: Modeling, Beamforming, \\and Experimental Validations}
		%
		%
		%
		
		\author{Lin Cao, Haifan Yin,~\IEEEmembership{Senior Member,~IEEE}, Li Tan, and Xilong Pei
			\thanks{This work was supported in part by the National Natural Science Foundation of China under Grants 62071191, 62071192, and 1214110.}
   			\thanks{The corresponding author is Haifan Yin.}
                \thanks{L. Cao, H. Yin, L. Tan and X. Pei are with School of Electronic Information and Communications, Huazhong University of Science and Technology,
430074 Wuhan, China  (e-mail: lincao@hust.edu.cn; yin@hust.edu.cn; ltan@hust.edu.cn;  pei@hust.edu.cn
).}}

	\maketitle
	
\begin{abstract}
Most research works on reconfigurable intelligent surfaces (RIS) rely on idealized models of the reflection coefficients, i.e., uniform reflection amplitude for any phase and sufficient phase shifting capability. In practice however, such models are oversimplified. This paper introduces a realistic reflection coefficient model for RIS based on measurements. The reflection coefficients are modeled as discrete complex values that have non-uniform amplitudes and suffer from insufficient phase shift capability. We then propose a group-based query algorithm that takes the imperfect coefficients into consideration while calculating the reflection coefficients. We analyze the performance of the proposed algorithm, and derive the closed-form expressions to characterize the received power of an RIS-aided wireless communication system. The performance gains of the proposed algorithm are confirmed in simulations. Furthermore, we validate the proposed theoretical results by experiments with our fabricated RIS prototype systems. The simulation and measurement results match well with the theoretical analysis.

\end{abstract}


%
\IEEEpeerreviewmaketitle

\section{Introduction}\label{Sec:introduction}
%
%
%
%
\IEEEPARstart {A}{s} the fifth-generation (5G) mobile communication gradually matured, the sixth-generation (6G) mobile communication has appeared on the horizon, which calls for much higher data rates, connection density, and energy efficiency.
With the expansion of the network scale and the ever increasing throughput requirement, mobile communications are facing challenges of high energy consumption and low cost efficiency\cite{wqq_ActivePassive_zy4}.

The recently proposed reconfigurable intelligent surface (RIS) provides new paradigms for wireless communication, for its potential of re-designing the wireless propagation environment and counteract the adverse radio conditions. As a result, RIS is actively being discussed as a prospective technology for 6G \cite{zy2}. 
RIS is a two-dimensional array of sub-wavelength elements that can be configured using a large number of passive components \cite{whatisRIS}. The passive electromagnetic response of the structure of each component (e.g., phase and amplitude) is controlled by simple programmable components, such as
positive-intrinsic-negative (PIN) diode \cite{PIN2, han_PIN}, varactor diode \cite{5G8Prototype, varactor1}, micro electro mechanical systems (MEMS) switch \cite{kappa_MEMS, manjappa_MEMS}, etc. By jointly manipulating these elements, RIS will be able to build a programmable wireless environment with low additional power or hardware expense \cite{6_9}, thereby further increasing the spectrum efficiency, energy efficiency, and physical layer security aspects of wireless communication systems.

To explore the potential of RIS techniques, RIS-aided wireless communication systems have recently been investigated in various applications/setups. Researchers propose new cost-effective solutions with RIS that lead to high beamforming gains and effective interference suppression using only low-cost reflecting elements, such as physical layer security \cite{physical_layer1}, orthogonal frequency division multiplexing (OFDM) \cite{OFDM1}, and integrated sensing and communications (ISAC) \cite{ISAC1}, etc.
The existing works are mostly based on the following three assumptions to make the system model more concise yet idealized:
\begin{itemize}
	\item Continuous phase shifts at its reflecting elements,
	\item Uniform reflection amplitude at any phase shift,
	\item Sufficient phase shifting capability covering the range from 0 to $2\pi$.
\end{itemize}

However, the RIS element with continuous phase shift is very challenging to implement due to its limited size and cost.
Take the RISs with varactor diodes or PIN diodes as examples, the RIS element with finely tuned phase shifts requires a wide range of biasing voltages of varactor diodes or a lot of PIN diodes and the corresponding controlling signals from the RIS controller.
As such, for practical RISs, it is more cost-effective to consider discrete phase shifts with a small number of control bits for each element.
Besides, due to hardware limitations\cite{hardware_limitation1,hardware_limitation2}, it is difficult and unrealistic to implement an RIS satisfying the ideal reflection model, which implies the reflection amplitude of each element is uniform at any phase shift.
The experimental results reported in \cite{5G8Prototype} show that the amplitude and phase shift of the reflected signal in a practical RIS system with varactor diodes are related to both the frequency of the incident signal and the biasing voltage of the varactor diode.
This is due to the fact that changing the frequency of incident signal or the control voltage will shift the equivalent impedance of the RIS element, leading to a variation in the ohmic loss in the system, which subsequently affects the amplitude and phase shift of the reflected signal.
In reality, this has long been a problem in RIS implementation\cite{longterm_problem}.
Moreover, the phase shifting capability is not always sufficient. The experimental results in \cite{5G8Prototype} show that the phase response of RIS element is sensitive to the angle of incident signal, which is due to the RIS being spatially dispersive\cite{DiscretePhase, angle2}. Besides, most existing RIS elements have limited phase shifting capability and cannot cover the range from 0 to $2\pi$.

There have been some previous studies to investigate the practical system model of RIS-aided systems \cite{ nidl1,nidl3, nidl_r, EnoughPhase3, nidl_r1, PracticalModelWuqq, nidl_r2,  nidl2, EnoughPhase1}. For instance, the works in \cite{nidl1,nidl3, nidl_r, EnoughPhase3, nidl_r1} consider various aspects of RIS-aided systems with discrete phase shift models with uniform reflective amplitude at each element. Specifically, the authors of \cite{ nidl1} introduce beamforming optimization techniques considering discrete phase shifts, while the work in \cite{nidl3} focuses on channel estimation and passive beamforming with discrete phase shifts. Furthermore, robust beamforming for the RIS-aided communication systems is investigated under imperfect channel state information (CSI) in \cite{ nidl_r }. Additionally, the authors of \cite{EnoughPhase3} explore rate-splitting multiple access for RIS-aided cell-edge users with discrete phase shifts. The work in \cite{ nidl_r1} evaluates the performance of uplink RIS-aided communication systems, particularly focusing on the impact of limited phase shifts on data rates. On the other hand, studies such as \cite{ PracticalModelWuqq, nidl_r2,  nidl2, EnoughPhase1} explore scenarios involving the practical system model that takes into account the intertwinement between the amplitude and phase response, with varying approaches to channel estimation and beamforming optimization techniques. For instance, the authors of \cite{ PracticalModelWuqq} propose a practical phase shift model and beamforming optimization for intelligent reflecting surfaces, highlighting the importance of realistic modeling in system design. Moreover, the authors of \cite{ nidl_r2} provide a comprehensive performance analysis based on the model proposed in \cite{ PracticalModelWuqq}. Furthermore, the authors of \cite{ nidl2} focus on practical modeling and beamforming techniques for wideband systems aided by RIS, while the channel estimation for practical IRS-assisted OFDM systems is addressed in \cite{ EnoughPhase1}, accounting for discrete phase shifts.
In the studies mentioned above, however, the authors primarily focused on the discrete phase shift at RIS elements, or the phase-dependent amplitude of the reflective signal, without considering the limited phase shifting range, which is a realistic problem in practice.

Motivated by the above, we propose a novel practical model for RIS-aided wireless communication systems, wherein discrete reflection coefficients and limited phase shifting capability are accounted for.
Unlike the models described in aforementioned studies, the proposed model imposes a restriction that the RIS elements can only induce a phase shift within a certain range, such as $[0,2\pi/3)$.
We formulate and solve an optimization problem to maximize the received power at the user with the constraint of non-ideal reflection coefficients.
The main contributions of this paper are as follows:

\begin{itemize}
	\item Based on experimental measurements, we introduce a realistic model for the reflection coefficients, which, to the best of our knowledge, is the first model taking into account the limited phase shifting capability.
	
	\item With the above model, we formulate a maximization problem for the received power, which is non-convex and difficult to solve. We propose a group-based query algorithm to find the solutions efficiently by calculating the corresponding phase range of each discrete reflection coefficient.
	
	\item We analyze the performance of the proposed algorithm. The closed-form expressions for the performances of RIS-aided communication systems are derived, including the cases of uniform reflection amplitude and non-uniform reflection amplitude.
	
	\item We conduct experiments with our fabricated RIS prototypes, to evaluate and validate the theoretical performance under practical deployment conditions. Two different RIS prototypes working on 5.8 GHz and 2.6 GHz are employed in the measurements. The experimental results match well with our theoretical analysis.
\end{itemize}

The rest of this paper is organized as follows.
Section \ref{Sec:model} introduces the system model for RIS-aided communication systems and derives the received power.
In Section \ref{Sec:Performance}, we propose a realistic reflection coefficient model and a group-based query algorithm for solving the received power maximization problem.
In Section \ref{Sec:Sim} we validate our theoretic results using both simulations and experimental measurements.
Section \ref{Sec:Conclusion} concludes the paper.

\section{System Model}\label{Sec:model}

We consider an RIS-aided wireless communication system as shown in Fig.\ref{fig:Channel_model}, where an RIS is adopted to reflect the signal from the base station (BS) towards the user. The RIS is composed of $M$ reflecting elements. By utilizing varactor diodes or PIN diodes, each element can shift the phase of the reflected signal. We show two examples of the structure of the element in Fig.\ref{fig:Channel_model}. A tunable phase shift to the reflected signal is achieved by varying the bias voltage of the diode.

\begin{figure}[!t]
	\centering
	\includegraphics[width=3in]{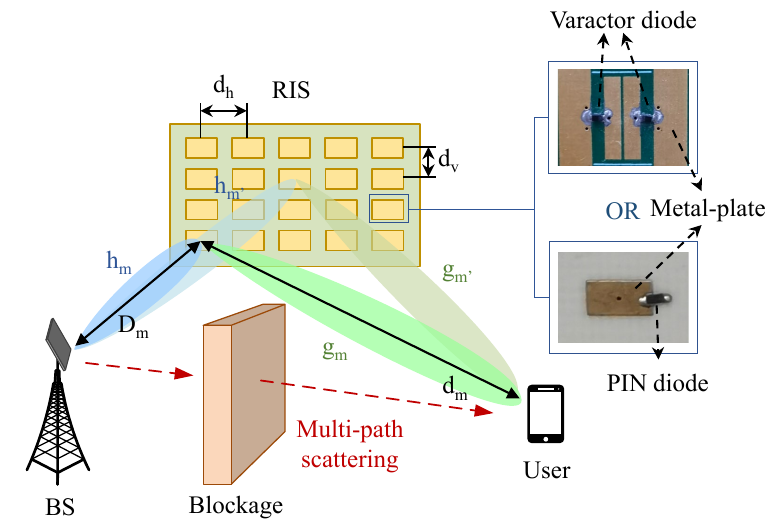}
	\DeclareGraphicsExtensions.
	\caption{An RIS-based wireless communication system.}
	\label{fig:Channel_model}
\end{figure}

For ease of exposition, we assume that there is no line-of-sight (LoS) path between the BS and the user. However, the following derivation can be easily extended for scenarios with LoS paths.
The signal received at the user is expressed as
\begin{equation}\label{eq1}
y = \sqrt {P_t} {\textbf{f}^H}\bm{\Phi }\textbf{h}s + n, 
\end{equation}
where $n\sim\mathcal{CN}\left(0{\rm{,}}\sigma^2 \right)$ is the additive white Gaussian noise (AWGN) with zero mean and variance of $\sigma ^2$, $P_t$ is the transmit power, $s$ is the transmitted signal with $\left | s^{2}  \right | = 1 $, and $\bm{\Phi } \buildrel \Delta \over={\rm diag}({A_1}{e^{ - j{\theta _1}}},\cdots,{A_M}{e^{ - j{\theta _M}}}) \in {\mathbb{C}^{M \times M}}$ denotes the reflection coefficient matrix of the RIS, where ${A_m}$ and ${\theta _m}$ are the reflection amplitude and phase shift on the incident signal, respectively. $\textbf{h}={\left[ {{h_1},\cdots,{h_M}} \right]^T \in {\mathbb{C}^{M \times 1}}}$ denotes the channel from the BS to the RIS, and $\textbf{f}={\left[ {{f_1},\cdots,{f_M}} \right]^T \in {\mathbb{C}^{M \times 1}}}$ denotes the channel from the RIS to the user.

Although the real direct path between the BS and the user is blocked, the LoS components exist in practical implementation due to the directed reflection of RIS. For this reason, the Rician fading is used to model the channels between the BS and the RIS, as well as the RIS and the user, as signified by $\textbf{h}$ and $\textbf{f}$, which are written as
\begin{equation}\label{eq2}
	\mathbf{h}= \sqrt {\frac{{K_1}}{{K_1} + 1}}\mathbf{\bar h}+\sqrt {\frac{1}{{K_1} + 1}}\mathbf{\tilde h },
\end{equation}
and
\begin{equation}\label{eq3}
	\mathbf{f}= \sqrt {\frac{{K_2}}{{K_2} + 1}}\mathbf{\bar f}+\sqrt {\frac{1}{{K_2} + 1}}\mathbf{\tilde f },
\end{equation}
where $\mathbf{\bar h}$ and $\mathbf{\bar f}$ are the LoS components of each channel; $\mathbf{\tilde h }$ and $\mathbf{\tilde f }$ are the non-LoS (NLoS) components; ${K_1}$ and ${K_2}$ are Rician K-factors of $\textbf{h}$ and $\textbf{f}$, respectively.

Since the distances between the BS and the RIS, as well as the RIS and the user, are significantly greater than the distances between any two RIS elements, we assume that the path loss of the BS-RIS link and the RIS-user link via different RIS elements is identical. The reflected LoS components of each channel via the $m$-th RIS element are denoted by \cite{Channel}
\begin{align}\label{eq4}
	\begin{split}
		\mathbf{\bar h } = &\sqrt {G_a{D_1^{-\alpha }}} \cdot \left[ e^{-j\frac{2\pi}{\lambda }D_1} {\rm{,}} \cdots {\rm{,}} 
	 e^{-j\frac{2\pi}{\lambda }D_m} {\rm{,}}\cdots {\rm{,}} e^{-j\frac{2\pi}{\lambda }D_M} \right]^T ,
	\end{split}
\end{align}
and
\begin{align}\label{eq5}
	\begin{split}
		\mathbf{\bar f } = &\sqrt {d_1^{-\alpha }}\left[ e^{-j\frac{2\pi}{\lambda }d_1} {\rm{,}}\cdots {\rm{,}} 
		e^{-j\frac{2\pi}{\lambda }d_m} {\rm{,}}\cdots {\rm{,}}e^{-j\frac{2\pi}{\lambda }d_M} \right]^T ,
	\end{split}
\end{align}
where $\alpha$ is the path loss factor, ${G_a}$ is the antenna gain, and $\lambda$ is the wavelength of the signal. $D_m$ and $d_m$ are the distances between the BS and the $m$-th RIS element and between the $m$-th RIS element and the user, respectively, as illustrated in Fig.\ref{fig:Channel_model}. During a channel coherent interval, the LoS components are constant, whereas elements of the NLoS components $\tilde{h}$ and $\tilde{f}$ are $i.i.d.$ (independent and identically distributed) complex Gaussian random variables\cite{Rician}.
The NLoS components of each channel are respectively denoted by
\begin{equation}\label{eq6}
	\mathbf{\tilde h} =L({D_1})\left[{g_1} {\rm{,}}\cdots{\rm{,}} {g_m}{\rm{,}} \cdots{\rm{,}}{g_M} \right]^T,
\end{equation}
and
\begin{equation}\label{eq7}
	\mathbf{\tilde f} =L({d_1})\left[ {b_1} {\rm{,}}\cdots{\rm{,}} {b_m}{\rm{,}} \cdots{\rm{,}} {b_M} \right]^T,
\end{equation}
where $L( \cdot )$ is the channel gain of the NLoS components, ${{g_m}\sim\mathcal{CN}\left(0{\rm{,}}1 \right)}$ and ${{b_m}\sim\mathcal{CN}\left(0{\rm{,}}1 \right)}$ denote the small-scale NLoS components.

We then derive the analytical expression for the maximum received power of the system. The instantaneous received power is given by
\begin{equation}\label{eq9}
	P_r =P_t{\left\| {{\mathbf{f}^H}\mathbf{\Phi}\mathbf{h}} \right\|^2},
\end{equation}

The instantaneous received power is a random variable.
The long-term average received power (LARP) $\Gamma$ is denoted by
\begin{equation}\label{eq10}
	\Gamma = \mathbb{E}\left\{P_r\right\} =  P_t \mathbb{E}\left\{{{{\left\| {{\mathbf{f} ^H}\mathbf{\Phi h}}\right\|}^2}} \right\}.
\end{equation}

To have a deeper understanding of the LARP $\Gamma$, we provide another form in the following proposition.
\begin{proposition}
	The LARP $\Gamma$ is given by
	\begin{align}\label{eq11}
		\begin{split}
			&\Gamma = \kappa _\text{NLoS}\sum\limits_m {A_{{m}}^2}+ \kappa _\text{LoS} \sum\limits_{m,m'} {{A_m}{A_{m'}}{e^{ - j[{\phi _m} - {\phi _{m'}} + {\theta _m} - {\theta _{m'}}]}}},
		\end{split}
	\end{align}	
	where ${\phi _m} = \frac{{2\pi }}{\lambda }({D_m} + {d_m})$ denotes the total phase shift induced by the LoS components of each channel; $\kappa _\text{LoS}$ and $\kappa_\text{NLoS}$ are constants defined as
	\begin{align}
		&\kappa _\text{LoS} = \frac{{{K_1}{K_2}{\eta _\text{LoS}}}}{{({K_1} + 1)({K_2} + 1)}},\\
		&\kappa_\text{NLoS} = \frac{K_1\eta _\text{NLoS1}+K_2\eta _\text{NLoS2}+\eta _\text{NLoS3}}{{({K_1} + 1)({K_2} + 1)}},
	\end{align}
	where $\eta _\text{LoS}$, $\eta _\text{NLoS1}$, $\eta _\text{NLoS2}$ and $\eta _\text{NLoS1}$ are constants related to the path loss of the channels, which are defined as 
	\begin{align}
		{\eta _\text{LoS}} &= {\sqrt {D_1^{-\alpha }d_1^{-\alpha }}}{P_t} {G_a} ,\\
		{\eta _\text{NLoS1}} &= {\sqrt {{G_a}D_1^{-\alpha }}}{P_t}{L( d_1 )},\\
		{\eta _\text{NLoS2}} &= {\sqrt {{G_a}d_1^{-\alpha }}}{P_t}{L( D_1 )},\\
		{\eta _\text{NLoS3}} &= {P_t}{L( D_1 )}{L( d_1 )}.
	\end{align}
	
\end{proposition}

\begin{proof}
	See Appendix \ref{Append: Prop1}.
\end{proof}

We aim to maximize the received power at the user by optimizing the response of each RIS element. The problem is formulated as
\begin{alignat}{3}
&& \mathrm{(P0)} :&\max_{\mathbf{\Phi}} {\quad}& {P_t}&\mathbb{E}\left\{{{{\left\| {{\mathbf{f} ^H}\mathbf{\Phi h}}\right\|}^2}} \right\}\\
&& &s.t.{\quad}&{\phi _m}& = {A_m}{e^{ - j{\theta _m}}}, m=1,\cdots,M,\nonumber\\
&& &&0 &\le {\theta _m} \le 2\pi, m=1,\cdots,M.\nonumber
\end{alignat}

The maximum LARP is obtained when ${\phi _m} - {\phi _{m'}} + {\theta _m} - {\theta _{m'}} = 0$ and ${{{{A}}_m} = 1}$ for any $m$ and $m'$. In other words, the optimal phase shifts with continuous value $\theta _m^*$ on the $m$-th RIS element should satisfy the following constraint:
\begin{equation}\label{eq13}
		\theta _m^* + {\phi _m} = C,
\end{equation}
where $C$ is an arbitrary constant. When the phase shift of each RIS element satisfies (\ref{eq13}) and all elements share the same amplitude value ${{{A}}_m} = 1$, the maximum LARP is obtained:
\begin{align}\label{eq14}
		\begin{split}
			 \Gamma _{\max }= &\mathop {\max }\limits_{ {\mathbf{\Phi } } } {P_t}\mathbb{E}\left\{{{{\left\| {{\mathbf{f} ^H}\mathbf{\Phi h}}\right\|}^2}} \right\}\\
			= &\kappa _\text{NLoS}M + \kappa _\text{LoS} M^2.
		\end{split}
\end{align}
${\Gamma _{\max }}$ will, therefore, serve as an upper bound of the received power.

According to (\ref{eq14}), the maximum LARP increases with the Rician K-factors of the channels.
The relationship between RIS size $M$ and the maximum LARP for different values of $K_1$ and $K_2$ is as follows.
When considering a pure LoS channel, i.e., ${K_1},{K_2} \to \infty $, an asymptotic squared maximum LARP of $O\left( {{M^2}} \right)$ can be achieved.
When considering a Rayleigh channel, i.e., ${K_1},{K_2} \to 0 $, an asymptotic linear LARP of $O\left( {{M}} \right)$ can be achieved.

\section{Performance Analysis for Practical System}\label{Sec:Performance}

Since continuous phase shifts are difficult to realize due to hardware limitations, discrete phase shifts are usually employed in practical systems. In this section, we will introduce a realistic discrete phase shifting model, and discuss how the employment of realistic discrete phase shifters affects the maximum LARP of an RIS-aided system.

\subsection{The realistic discrete phase shifting model}
We first assume that the RIS has a phase shifting capability of $\omega$, which means it can generate a phase shift covering the range from $0$ to $\omega$, and that the phase shift is $k$-bit uniformly quantized.
In other words, we control the programmable components such as varactor diodes or PIN diodes, to generate ${{\rm{2}}^k}$ patterns of the reflection coefficients, which are denoted by:
 \begin{equation}\label{eq18}
 	{\Phi _i} = A_{\theta _{{i}}}{e^{ - j{\theta _{{i}}}}},\quad i = 1,2,\cdots,{2^k},
 \end{equation}
where $A_{\theta _{{i}}}$ is the amplitude when the phase shift is ${{\theta _{{i}}}}$.

To investigate the effect of phase shifting capability on performance, RIS systems are divided into two categories: systems with sufficient phase shifting capability and systems with insufficient phase shifting capability.
We consider phase shifting capability to be sufficient when it can meet the quantification requirements. For example, the phase shifting capability in a 1-bit quantized system surpasses 180° and in a 2-bit quantized system exceeds 270°, etc.
Otherwise, the RIS system is called an insufficient phase shifting capability system.

For systems having sufficient phase shift capability, i.e., $\omega\ge\frac{{2^k} - 1}{2^k} 2\pi$, the uniform phase interval is $\Delta \theta  = \frac{{2\pi }}{{{2^k}}}$. For systems with insufficient phase shift capability, i.e., $\omega  < \frac{{{2^k} - 1}}{{{2^k}}} 2\pi$, the uniform phase interval after quantization is $\Delta \theta  = \frac{\omega }{{{2^k} - 1}}$. ${{\theta _{{i}}}}$ is given by
\begin{equation}\label{eq19}
	\theta _i =\left\{
	\begin{aligned}
		&\frac{\omega}{{2^k} - 1} \cdot \left( i - 1 \right),\quad \omega  < \frac{{2^k} - 1}{2^k} \cdot 2\pi; \\
		&\frac{{2\pi }}{{{2^k}}} \cdot \left( i - 1\right),\quad \omega  \ge \frac{{2^k} - 1}{2^k} \cdot 2\pi .\\
	\end{aligned}
	\right.
\end{equation}  

We will analyze the performance of an RIS-aided system employed the realistic phase shifter with or without the assumption of ideal reflection, respectively, in the following subsection.

\subsection{Analysis on the realistic discrete phase shifting model with uniform amplitude}

In this subsection, we will discuss the impact of limited phase shift range on the maximum LARP, using the ideal reflection model with uniform reflective amplitude, which means each discrete phase shift $\theta _{i}$ corresponds to an amplitude of $A_{\theta_i}=1$. In this scenario, the problem $\mathrm{ (P0)}$ of maximizing the received power is transformed to
\begin{alignat}{3}
&& \mathrm{(P1)} :&\max_{\mathbf{\hat{\Phi}}} {\quad}& {P_t}&\mathbb{E}\left\{{{{\left\| {{\mathbf{f} ^H}\mathbf{\Phi h}}\right\|}^2}} \right\}\\
&& &s.t.{\quad}&{\hat{\phi }_m}& = {e^{ - j{\hat{\theta} _m}}}, m=1,\cdots,M,\label{P1 constr1}\\
&& &&{\hat{\theta}_m} &\in (\theta_1,\cdots,\theta_i,\cdots,\theta_{2^k}), m=1,\cdots,M.\label{P1 constr2}
\end{alignat}

Under the constraints in (\ref{P1 constr1}) and (\ref{P1 constr2}), the LARP of the RIS-aided system is written as
\begin{align}\label{eqP1}
	\begin{split}
		\Gamma = &\kappa _\text{NLoS}M + \kappa _\text{LoS}\sum\limits_{m,m'} {e^{-j[({\phi _m}+{\theta _m})-({\phi _{m'}}+{\theta _{m'}})]}}. \\
	\end{split}
\end{align}

To obtain the maximum LARP in this scenario, for the $m$-th RIS element, we select the discrete phase shift which is closest to the optimal one $\theta _m^*$ as given in (\ref{eq13}), and denote it by ${\hat \theta _m}$. The phase errors resulting from discrete phase shifts are defined as
\begin{equation}\label{eq15}
	{\delta _m} = \theta _m^* - {\hat \theta _m}, m = 1,\cdots,M.
\end{equation}

The maximum LARP $\hat \Gamma_{\rm{max}}$ with discrete phase shifts is given by
\begin{align}\label{eq16}
	\begin{split}
		&\hat \Gamma_{\rm{max}}=\kappa _\text{NLoS}M +\kappa _\text{LoS}\sum\limits_{m,m'} {e^{-j\left[C + {\delta _m} - C - {\delta _{m'}}\right]}} \\
		&=\kappa _\text{NLoS}M+ \kappa _\text{LoS} \sum\limits_{m,m'}{\left({\sin{\delta_m}\sin{\delta_{m'}} + \cos{\delta_m}\cos{\delta_{m'}}}\right)} .
	\end{split}
\end{align}

Since ${\phi _m}$ is jointly determined by the wavelength of the incident signal, the distance between the BS and the $m$-th RIS element, and the distance between the $m$-th element of the RIS and the user, and that $\theta _m^* + {\phi _m} = C$, we assume that the optimal phase shift $\theta _m^*$ in the practical system is uniformly distributed in $[0,2\pi )$.
The following theorem shows the expectation of the maximum LARP $\mathbb{E}(\hat \Gamma_{\rm{max}})$ in RIS-aided wireless communication systems in this scenario.

\begin{theorem}\label{ThmIdl}
	Assuming that all elements of RIS have the same reflection amplitude $A$, the closed-form expression for the expectation of the maximum LARP at the user is given by 
	\begin{align}\label{eq17}
		&\mathbb{E} \left (\hat \Gamma_{\rm{max}} \right ) =\nonumber\\
		&\left\{
		\begin{aligned}
			&\kappa _\text{NLoS}M +4\kappa _\text{LoS}M^2 \\
			&\times\left[{P_1}\sin b + {P_2}(\sin a - \sin b) \right]^2, \quad \omega<\frac{2^{{k}}-1}{2^{{k- 1}}}\pi ;\\
			&\kappa _\text{NLoS}M + \kappa _\text{LoS}M^2\cdot \frac{2^{{2k}}}{\pi ^2}{\sin ^2}\frac{\pi}{2^{{k}}}, \quad \omega  \ge \frac{2^{{k}}-1}{2^{{k- 1}}}\pi.
		\end{aligned}
		\right.
	\end{align}
	where $a = \pi - \frac{\omega }{2}$, $b  = \frac{\omega }{{2\left({{2^{{k}}} - 1} \right)}}$, ${P_1} = \frac{{{2^{{k}}}}}{{2\pi }}$, ${P_2} = \frac{1}{{2\pi }} $.
\end{theorem}

\begin{proof}
	See Appendix \ref{Append: Theor1}.
\end{proof}

Theorem \ref{ThmIdl} indicates that the maximum LARP expectation $\mathbb{E}\left(\hat\Gamma_{\rm{max}}\right)$ is determined by the size and topology of the system, the propagation environment, the number of quantization bits $k$ and the phase shift capability $\omega$. We will further show in the next section that the phase shift capability $\omega$ is a key factor for performance.

 \subsection{Analysis on the realistic discrete phase shifting model}
 
 In practical systems, the amplitude response of RIS reflecting elements generally depends on the phase shift value. In this section, we further consider the influence of phase shifting capability on the maximum LARP based on the non-ideal reflection model in which the reflection amplitude varies with the phase shift. In this scenario, based on the constraint of the realistic discrete phase shifting model, the problem of maximizing LARP is written by
\begin{alignat}{3}
&& \mathrm{(P2)} :&\max_{\mathbf{\hat{\Phi}}} {\quad}& {P_t}&\mathbb{E}\left\{{{{\left\| {{\mathbf{f} ^H}\mathbf{\Phi h}}\right\|}^2}} \right\}\label{P2}\\
&& &s.t.{\quad}&{\hat{\phi}_m}& = {A_{\hat{\theta} _m}}{e^{ - j{\hat{\theta} _m}}}, m=1,\cdots,M,\label{P2 constr1}\\
&& &&{\hat{\theta} _m} &\in (\theta_1,\cdots,\theta_i\cdots,\theta_{2^k}), m=1,\cdots,M.
\end{alignat}

Although the objective function of (P2) is convex in this scenario, solving (P2) is difficult due to the non-convex constraint in (\ref{P2 constr1}). When using the non-ideal reflection model, the reflection design should strike an appropriate balance between the amplitude and phase of the reflected signal. To solve this problem, we propose a low-complexity algorithm based on vector quantization of reflection coefficients to find an approximate solution to (P2).

We begin by defining a quantization loss function $\mathcal{L}$ to assess the difference between the quantified and desired reflection coefficients:
\begin{equation}\label{eq20}
\mathcal{L}(\theta_i,m) = 1 - {A_{\theta_i}}\cos ({\theta _{{i}}} - \theta^*_m ),
\end{equation}
which can be thought of as the difference between the desired reflection coefficient and the quantified reflection coefficient projected onto it. Here, the desired reflection coefficient denotes the reflection coefficient that maximizes the LARP without any constraints, with a reflection amplitude $A_m = 1$ and a phase shift $\theta^*_m$ given in (\ref{eq13}). We will refer to $\theta^*_m$ as the expected phase shift in the remainder of this paper. The optimization problem of the reflection coefficient of the $m$-th element is then simplified to (P3):
\begin{alignat}{3}
&& \mathrm{(P3)} :&\max_{\hat{\phi }_m} {\quad}& a_\text{NLoS}&A_{\hat{\theta} _m}^2 + {a_\text{LoS}}{A_{\hat{\theta }_m}}\cos (\hat{\theta }_m - \theta^*_m)\label{P3}\\
&& &s.t.{\quad}&{\hat{\theta} _m} &\in (\theta_1,\cdots,\theta_i,\cdots,\theta_{2^k}), m=1,\cdots,M.\label{P3 constr1}
\end{alignat}
where $a_\text{NLoS}$ and $a_\text{LoS}$ denote the constants related to the channel path loss as
\begin{align}
    {a_\text{NLoS}} &= \frac{K_1\eta _\text{NLoS1}+K_2\eta _\text{NLoS2}+\eta _\text{NLoS3}}{{({K_1} + 1)({K_2} + 1)}},\\
    {a_\text{LoS}} &=M\bar{A}  \frac{{{K_1}{K_2}{\eta _\text{LoS}}}}{{({K_1} + 1)({K_2} + 1)}}.
\end{align}
Here, $\bar{A} = \sum{{A_{\hat{\theta} _i}^2}}/\sum{{A_{\hat{\theta} _i}}}$ is a constant related to the power loss caused by the reflection coefficient of RIS, which is derived based on the characteristic that the optimized phase shift is usually more concentrated towards the phase shift with the larger reflective amplitude\cite{PracticalModelWuqq}.

Note that (P3) can be solved by the exhaustive search method. Since in practical systems, the number of control bits $k$ is generally not greater than 3, the complexity of this method is not exceptionally high.
Moreover, since each RIS element can generate the same patterns of reflection coefficients, any RIS element with the identical expected phase shift has the same optimal reflection coefficient when solving the problem (P3).
Therefore, we may build a look-up table by calculating the expected phase shift range $c_i$ for each quantized reflection coefficient. This table provides the optimized reflection coefficients for any possible value of expected phase shifts. In other words, the reflection coefficients are quantized using the look-up table, which will further reduce the computational complexity of solving (P3).  
Below we will show how this table is developed. For notational simplicity, we omit the subscript $m$ of the expected phase $\theta_m^*$.

Firstly, substitute each quantized reflection coefficient into the objective function (\ref{P3}) and obtain a series of equations as follows:
\begin{align}\label{Eq:objFuns}
	&{f_1}(\theta ) = {a_\text{NLoS}}A_{{\theta _1}}^2 + {a_\text{LoS}}{A_{{\theta _1}}}\cos ({\theta _1} - \theta),\nonumber\\
	&\cdots\nonumber\\
	&{f_i}(\theta) = {a_\text{NLoS}}A_{{\theta _i}}^2 + {a_\text{LoS}}{A_{{\theta _i}}}\cos ({\theta _i} - \theta),\\
	&\cdots\nonumber\\
	&{f_{2^k}}(\theta) = {a_\text{NLoS}}A_{{\theta _{{2^k}}}}^2 + {a_\text{LoS}}{A_{{\theta _{{2^k}}}}}\cos ({\theta _{{2^k}}} - \theta).\nonumber
\end{align}

 \begin{figure}[!t]
 	\centering
 	\includegraphics[width=2.5in]{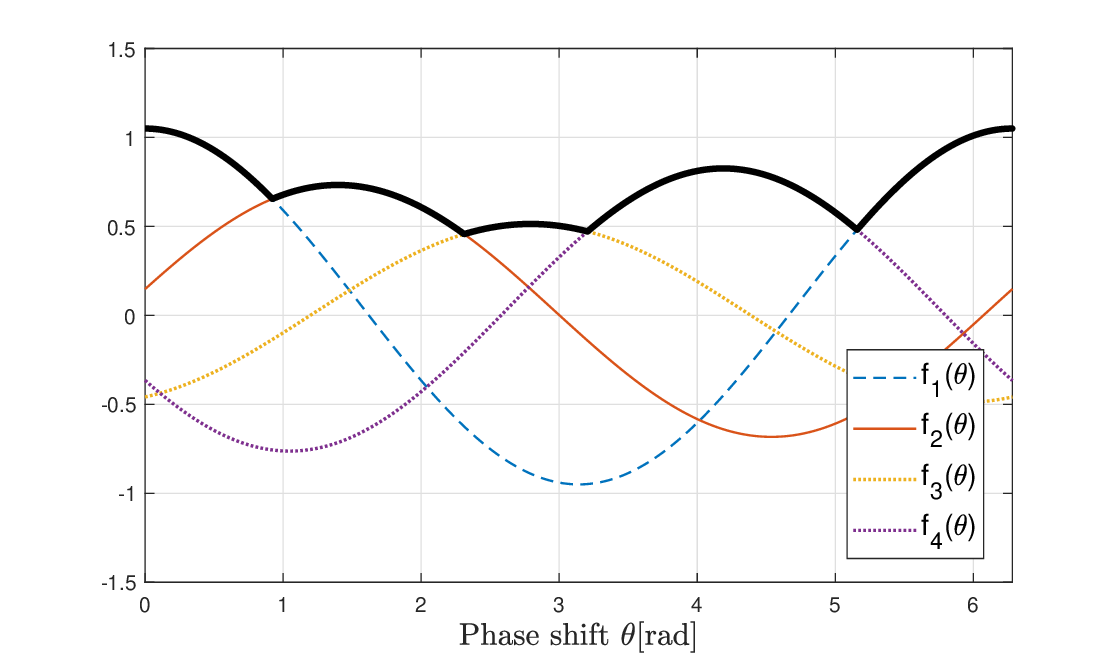}
 	\DeclareGraphicsExtensions.
 	\caption{An example of a set of curves of a 2-bit quantized RIS-aided system.}
 	\label{fig_CurveCluster}
 \end{figure}

As shown in Fig. \ref{fig_CurveCluster}, the solution to problem (P3) when the expected phase $\theta^* = \theta$ is the reflection coefficient corresponding to the envelope of the set of curves at phase shift $\theta$.
As a result, determining the expected phase range $c_i$ is equivalent to finding among the equations (\ref{Eq:objFuns}) that makes the objective function maximum, and calculating the value range of $\theta$.

In order to solve this problem, we can start by seeking the intersection of every two curves in [0,2$\pi$). For instance, for ${f_i}(\theta )$ and ${f_{i'}}(\theta )$, we have
\begin{align}\label{eq fifi1}
    &{a_\text{NLoS}}A_{{\theta _i}}^2 + {a_\text{LoS}}{A_{{\theta _i}}}\cos ({\theta _i} - \theta )\nonumber \\ 
    &= {a_\text{NLoS}}A_{{\theta _{i'}}}^2 + {a_\text{LoS}}{A_{{\theta _{i'}}}}\cos ({\theta _{i'}} - \theta ).
\end{align}
The above formula can be converted into
\begin{align}\label{eq fifi2}
    \sqrt {{C_{\sin}^2} + {C_{\cos}^2}}\sin (\theta  \pm \vartheta )= \frac{{{a_\text{NLoS}}(A_{{\theta _{i'}}}^2 - A_{{\theta _i}}^2)}}{{{a_\text{LoS}}}},
\end{align}
where $C_{\sin}$ and $C_{\cos}$ are constants, which are determined by
\begin{align}\label{eq csincos}
	&C_{\sin}={A_{{\theta _i}}}\sin {\theta _i} - {A_{{\theta _{i'}}}}\sin {\theta _{i'}},\\ 
	&C_{\cos}={A_{{\theta _i}}}\cos {\theta _i} - {A_{{\theta _{i'}}}}\cos {\theta _{i'}}.
\end{align}
The quadrant of $\vartheta$ is defined by the sign of $C_{\sin}$ and $C_{\cos}$, and $\vartheta$ is defined by $\tan\vartheta = {C_{\cos}}/{C_{\sin}}$. Taking $C_{\sin}>0$ as an example, we solve the intersection of ${f_i}(\theta )$ and ${f_{i'}}(\theta )$ as
\begin{align}\label{eq thetaii}
{\theta _{ii'}} = \arcsin \frac{{{a_\text{NLoS}}(A_{{\theta _{i'}}}^2 - A_{{\theta _i}}^2)}}{{{a_\text{LoS}}\sqrt {{C_{\sin }}^2 + {C_{\cos }}^2} }} - \arctan \frac{C_{\cos}}{C_{\sin}}
\end{align}

Then we need to determine if the following equality holds:
\begin{align}\label{eq thetaii_max}
    {f_i}(\theta _{ii'}) = \max [{f_1}(\theta _{ii'}),\cdots,{f_i}(\theta _{ii'}),\cdots,{f_{{2^k}}}(\theta _{ii'})].
\end{align}
If it does, we call it a valid intersection. Following that, we compare the values of ${f_i}'(\theta _{ii'})$ and ${f_{i'}}'(\theta _{ii'})$. If ${f_i}'(\theta _{ii'})>{f_{i'}}'(\theta _{ii'})$, the phase shift range between $\theta _{ii'}$ and the next valid intersection belongs to $c_i$, and the phase shift range between the last valid intersection and $\theta _{ii'}$ belongs to $c_{i'}$, and vice versa. By calculating all the valid intersections between the curves and comparing the derivatives of the corresponding curves at valid intersections, the expected phase range corresponding to each quantized state is obtained.

Once the phase range $c_i$ is obtained,
we may easily compute the quantized coefficient with the expected phase shift, which will greatly reduce the computational complexity.
The overall procedure to solve (P2) is summarized in Algorithm \ref{alg:A} which is referred to as group-based query algorithm.

Next, we will give a brief discussion on the computational complexity of the group-based query algorithm.
The presented algorithm can be divided into two parts: grouping and querying. 
The grouping part mainly contains comparisons and inverse trigonometric functions. The approximate complexity of this part is ${o_1} = O\left( {rt} \right) + O\left( {{2^k}r} \right) + O\left( r \right)$, where $r = 2^k(2^k-1)$ is the number of iterations of the loop and $t$ denotes the accuracy for the inverse trigonometric functions. Since $k$ is generally not greater than 3, the complexity of this part is very low.
The approximate complexity of the querying part is ${o_2} = O\left( M \right)$.
Therefore, when $M\gg k$, the total complexity of the algorithm is approximated as ${o_1} + {o_2} \approx O\left( M \right)$. Otherwise, when $k$ is large, then the total complexity of the algorithm is approximated as ${o_1} + {o_2} \approx O\left( {{2^k}r} \right) + O\left( M \right)$.

\begin{algorithm}[t]
\caption{Group-Based Query Algorithm for Solving (P2)}  
\label{alg:A}  
\begin{algorithmic}[1]  
\REQUIRE ${a_\text{LoS}}$, ${a_\text{NLoS}}$, $M$, $\left\{ {\theta _m^*} \right\}_{m = 1}^M$, $\left\{ A_{\theta_i} \right\}_{i = 1}^{2^k}$, $\left\{ \theta_i \right\}_{i = 1}^{2^k}$. 
\ENSURE $\mathbf{\hat{\Phi}}$
\FOR{$i = 1$ to $2^k-1$}
    \FOR{$i' = i+1$ to $2^k$}
    \STATE Solve (\ref{eq fifi1}) by using (\ref{eq fifi2}), and obtain the solution ${\theta_{ii'}}$ in $[0, 2\pi)$.
    \IF{$ {f_i}(\theta _{ii'}) = \max [{f_1}(\theta _{ii'}),\cdots,{f_i}(\theta _{ii'}),\cdots,{f_{{2^k}}}(\theta _{ii'})]$}
    \STATE Compare the values of ${f_i}'(\theta _{ii'})$ and ${f_{i'}}'(\theta _{ii'})$ to get the expected phase range $c_i$ and $c_{i'}$ corresponding for each quantized reflection coefficient.
    \ENDIF
    \ENDFOR
\ENDFOR
\STATE According to all valid intersections and their derivatives, the expected phase range $c_i$ is obtained.
\FOR{$m = 1$ to $M$} 
\STATE Find $\hat{\theta}_m = \theta_{i^*}$ as the solution to (P3) with the expected phase shift range $c_i$.
\STATE ${\hat{\phi}_m} = {A_{\hat{\theta} _m}}{e^{ - j{\hat{\theta}_m}}}$.
\ENDFOR      
\end{algorithmic}  
\end{algorithm}

Below we will analyze the performance of the proposed algorithm and derive the expectation of the maximum LARP $\mathbb{E}\left(\hat\Gamma_{\rm{max}}\right)$ at the user in RIS-aided wireless communication systems.
\begin{theorem}\label{ThmNIdl}
	Assuming that each RIS element can produce $2^k$ discrete reflection coefficients as defined in (\ref{eq18}), the expectation of the maximum LARP $\mathbb{E}\left(\hat\Gamma_{\rm{max}}\right)$ at the user is given by
	\begin{align}\label{eq21}
		\begin{split}
			&\mathbb{E}\left(\hat\Gamma_{\rm{max}}\right)=\frac{M\kappa _\text{NLoS}}{{2{\pi}}}\sum\limits_{i} {\frac{\mu_i}{2\pi}A_{\theta _i}^2 }+\frac{{{M^2}\kappa _\text{LoS}}}{{{4{\pi ^2}}}}\sum\limits_{i,i'}A_{\theta _i}A_{\theta _{i'}} \\
			&\times \int_{{\delta _m} \in {d_i}} {\int_{{\delta _{m'}} \in {d_{i'}}} {(\sin {\delta _m}\sin {\delta _{m'}} + \cos {\delta _m}\cos {\delta _{m'}})d{\delta _m}d{\delta _{m'}}} }. \\
		\end{split}
	\end{align}
	where $d_i$ is the quantization error of the i-th quantized reflection coefficient $\Phi_i$.
\end{theorem}
\begin{proof}
  See Appendix \ref{Append: Theor2}.
\end{proof} 
 
 The proposed algorithm shows that the expected phase shift range ${c_i}$ is related to the parameters of the RIS-aided system, such as the phase shifting capability $\omega$, the reflective amplitude $A_{\theta_i}$, etc.
 Therefore, whenever we compute the expectation of the maximum LARP $\mathbb{E}\left(\hat\Gamma_{\rm{max}}\right)$ of different systems, we need to recalculate the $c_i$ according to the parameters of the system.
 However, $c_i$ has a simpler solution when the number of bits $k=1$ and the pure LoS path, which means the power of the RIS-reflected signal dominates in the total received power. In other words, the Rician K-factors of the channels ${K_1},{K_2} \to \infty $. In this case, the expectation of the maximum LARP $\mathbb{E}\left(\hat\Gamma_{\rm{max}}\right)$ at the user is shown in Corollary \ref{PropNIdl}.
 
 \begin{corollary}\label{PropNIdl}
 	Assuming that each RIS element can produce 2 discrete reflective coefficients as ${\Phi _1} = A_{\theta_1}{e^{ - j{\theta _1}}} $ and ${\Phi _2}= A_{\theta_2} {e^{ - j{\theta _2}}} $, and that the channels contain pure LoS paths. The optimal phase shift ranges are as follows:
 	\begin{equation}\label{ci}
 		\begin{aligned}
 			&{c_1} \in \left[ {0,{\psi _1}} \right) \cup \left[ {{\psi _2},2\pi } \right),\\
 			&{c_2} \in \left[ {{\psi _1},{\psi _2}} \right),
 		\end{aligned}
 	\end{equation}
    where ${\psi _1} =  - \arctan \frac{{{A_{\theta_2}}\cos {\omega '} - {A_{\theta_1}}}}{{A_{\theta_2}}{\sin {\omega '}}} \in \left[ 0, \frac{\pi}{2} \right)\cup \left( \frac{\pi}{2}, \pi \right)$ , ${\psi _2} = \pi  - \arctan \frac{{{A_{\theta_2}}\cos {\omega '} - {A_{\theta_1}}}}{{A_{\theta_2}}{\sin {\omega '}}} \in \left[ \pi, \frac{3\pi}{2} \right)\cup \left( \frac{3\pi}{2}, 2\pi \right)$.
    
 	The closed-form expression for the expectation of the maximum LARP  $\mathbb{E}\left(\hat\Gamma_{\rm{max}}\right)$ at the user is given by
 	\begin{align}\label{eq22}
 		&\mathbb{E}\left(\hat\Gamma_{\rm{max}}\right)\nonumber \\
 		&=
 		\left\{
 		\begin{aligned}
 			&\frac{{{\eta _\text{LoS}}{M^2}}}{{\pi ^2}}\left[A_{\theta_1}^2 + A_{\theta_2}^2 - 2A_{\theta_1}A_{\theta_2}\cos \omega \right],\omega  < \pi;\\
 			&\frac{{{\eta _\text{LoS}}{M^2}}}{{\pi ^2}}\left[A_{\theta_1}^2 + A_{\theta_2}^2 + 2{A_{\theta_1}}{A_{\theta_2}}\right],\omega  \ge \pi .\\
 		\end{aligned}
 		\right.
 	\end{align}
 \end{corollary}

 \begin{proof} 
 See Appendix \ref{Append: Cor1}.
\end{proof}

The above corollary shows that given the amplitude for each quantized reflection coefficient, the LARP of the RIS-aided system with the number of bits $k = 1$ tends to decrease when the phase shift capability $\omega$ declines. However, provided a phase shift capability $\omega$, the impact of the reflection amplitude on the RIS-aided system is not straightforward; we will show this in the following section using simulations.

 \section{Simulation and Experimental Results}\label{Sec:Sim}
 
In this section, we make both simulations and experimental measurements to evaluate the performance of the group-based query algorithm and validate the theoretic results presented in this work.
 
 \subsection{Simulation Results}

\begin{figure}[!t]
	\centering
	\includegraphics[width=3in]{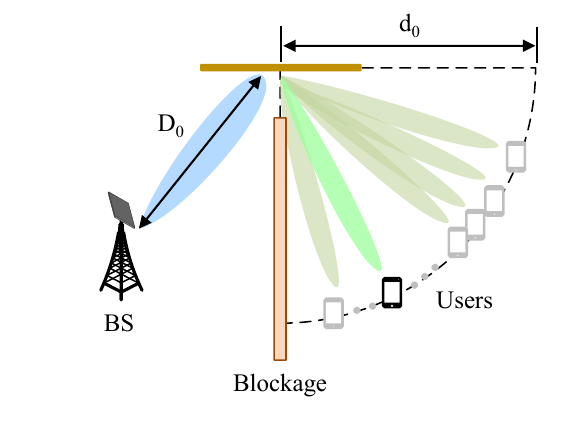}
	\DeclareGraphicsExtensions.
	\caption{Simulation setup.}
	\label{FigSim}
\end{figure}

\begin{figure}[!t]
	\centering
	\includegraphics[width=3in]{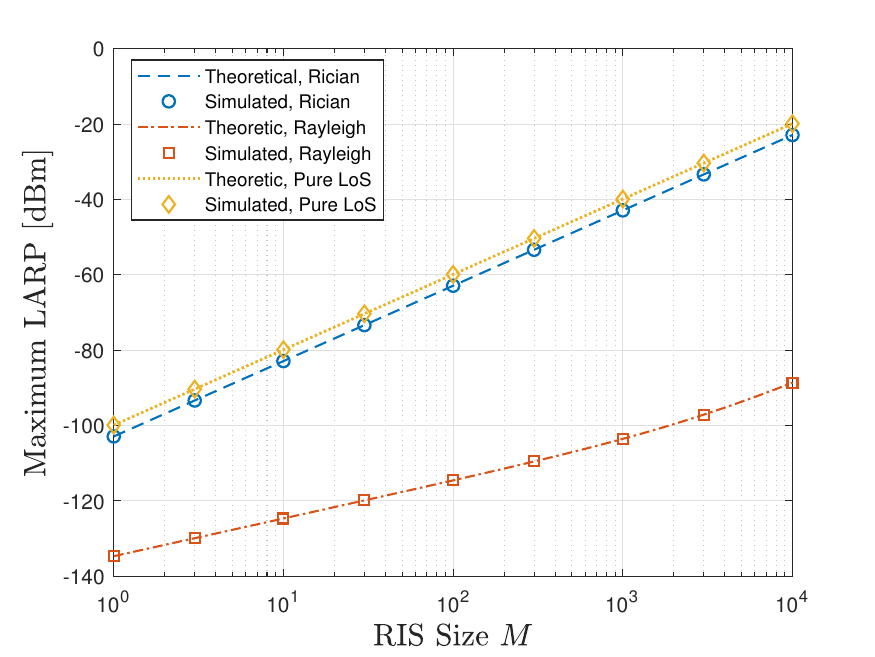}
	\DeclareGraphicsExtensions.
	\caption{The LARP versus the RIS size $ M $, continuous phase shifts.}
	\label{fig_sim}
\end{figure}

In this subsection, we analyze the performance of an RIS-aided communication system with users distributed randomly and uniformly on a quarter circle of radius $d_0$ centered on the RIS, as shown in Fig. \ref{FigSim}. The simulation results in all figures are averaged over 2000 independent realizations of the different user locations. The channel parameters and RIS system parameters were chosen in accordance with the 3GPP propagation environment outlined in \cite{3gpp}. Unless otherwise notified, the simulation parameters are as follows. ${D_0 = 90 \; \text{m}}$ is the distance between the BS and the RIS center, and ${d_0 = 70 \; \text{m}}$ is the distance between the user and the RIS center. The number of RIS elements is $ {M  = 4096}$, and the sizes of RIS elements are ${d_h = d_v = 0.05 \; \text{m}}$. The transmit power is ${20 \; \text{dBm}}$, and the noise power is ${-90 \; \text{dBm}}$. The path loss of LoS and NLoS is configured based on the UMa model defined in \cite{3gpp}. The carrier frequency is ${2.6 \; \text{GHz}}$, and the Rician factor is $K_1 = K_2 = 4$.

Fig. \ref{fig_sim} shows the maximum LARP with continuous phase shifts versus the number of elements $ M $. As shown in Fig. \ref{fig_sim}, our theoretical results are very close to the simulated ones. The figure also shows that for all three channel conditions, the maximum received power increases with the number of elements $ M $. Furthermore, the slope of the maximum received power curve for the pure LoS channel is 20, indicating that the received power is proportional to the square of the number of RIS elements $ M $. Similarly, the slope of the maximum LARP curve for the Rayleigh channel is 10, meaning that the received power is proportional to $ M $. In addition, we can see that the maximum LARP increases with the Rician factors $K_1$ and $K_2$.

\begin{figure}[!t]
	\centering
	\includegraphics[width=3in]{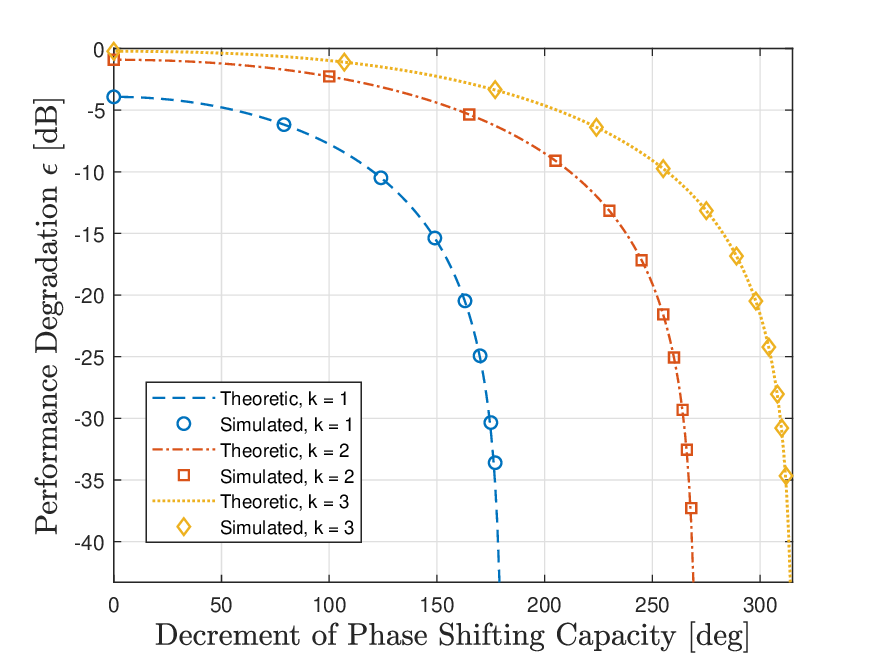}
	\DeclareGraphicsExtensions.
	\caption{The performance degradation $\varepsilon$ versus the decrement of phase shifting capability $c$, uniform reflection amplitude.}
	\label{Fig5}
\end{figure}

To quantify the performance degradation in this scenario, we define a loss factor $\varepsilon$:
\begin{equation}\label{eq23}
	{\varepsilon}={\log _{10}}\frac{\mathbb{E}\left(\hat\Gamma_{\rm{max}}\right)}{{{\Gamma _{\rm{max}}}}}.
\end{equation}
The loss factor is the ratio of the average received power based on the realistic discrete phase shifting model to the maximum received power with continuous phase shift given by (\ref{eq14}), which represents the performance degradation caused by practical phase shifters.

Fig. \ref{Fig5} shows the expectation of the LARP as a function of the decrement of phase shifting capability $c$ under the ideal reflection model.
Here, $c$ is defined as the difference between the phase capability of the RIS and the necessary phase shifting capability of the quantization bits, which means $c=\max [0, \frac{{2^k} - 1}{2^k} \cdot 2\pi-\omega]$.
In Fig. \ref{Fig5}, the numbers of quantization bits are set to 1, 2, and 3, respectively.
As can be seen, more quantization bits lead to higher performance for the same phase shifting capability, which is expected because a larger number of bits reduces phase quantization error.
Furthermore, according to Fig. \ref{Fig5}, for systems with different numbers of control bits, the LARP decreases significantly with the decrement of the phase shifting capability of the system, especially when the number of quantization bits is smaller.
For the 1-bit, 2-bit, and 3-bit quantized reflection coefficients, a 3 dB LARP degradation is caused by a 90°, 140°, and 175° phase capability decrement, respectively.
Besides, the theoretical results obtained according to Theorem \ref{ThmIdl} are in good agreement with the simulation results in Fig. \ref{Fig5}.

\begin{figure}[!t]
	\centering
        \includegraphics[width=3 in]{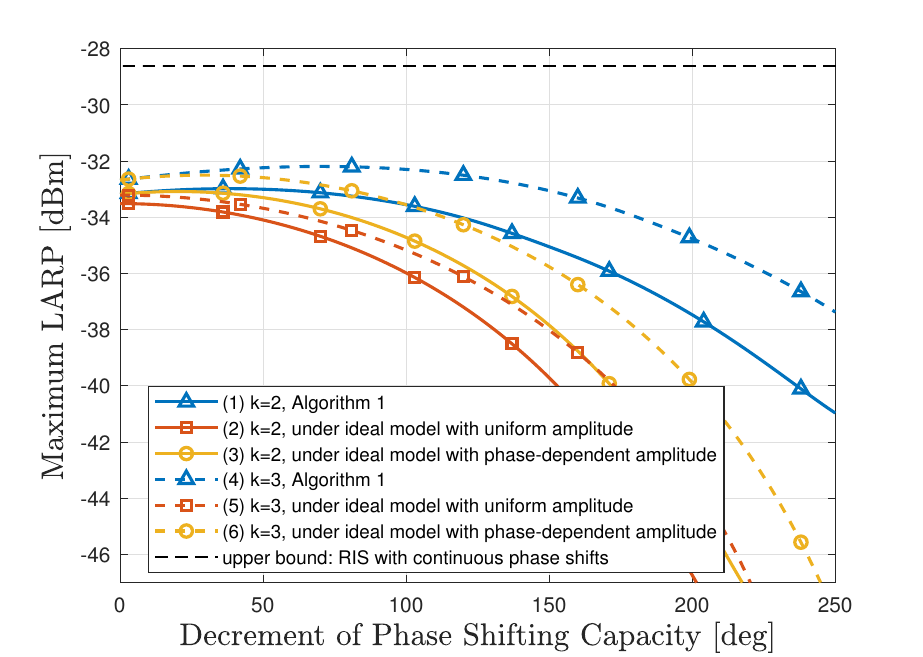}
        \DeclareGraphicsExtensions.
        \caption{The maximum LARP versus decrement of phase shifting capability $c$ when $k=2$ and $k=3$.}
        \label{Fig:AlgA}
\end{figure}

\begin{figure}[!t]
	\centering
        \includegraphics[width=3 in]{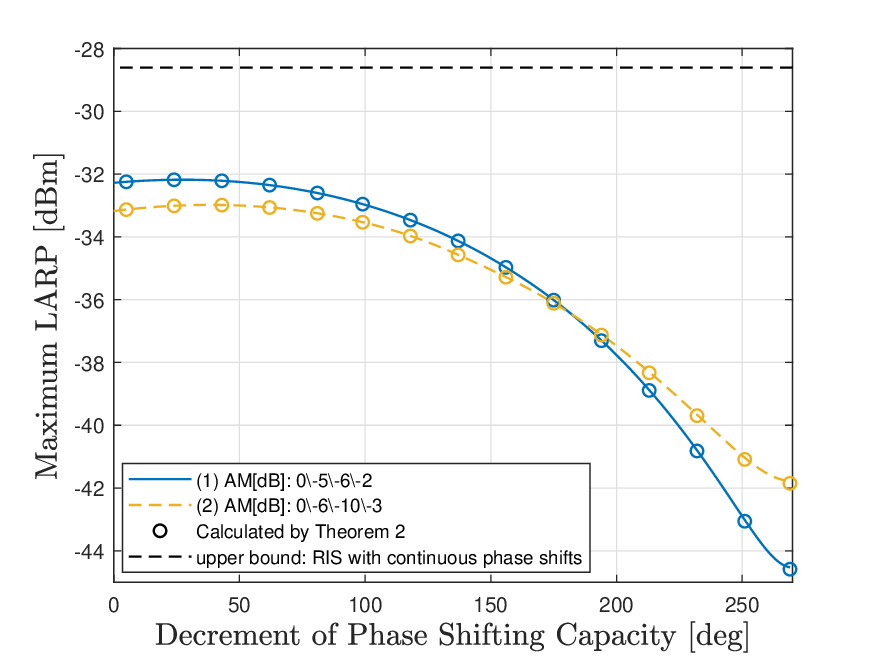}
        \DeclareGraphicsExtensions.
        \caption{The maximum LARP versus decrement of phase shifting capability $c$ under different amplitudes of the reflected signal when $k=2$.}
        \label{Fig:AlgB}
\end{figure}

Next, by varying the decrement of phase shifting capability $c$, the LARP is compared in Fig. \ref{Fig:AlgA} for the following three schemes of computing the discrete RIS phase shifts: (i) group-based query algorithm, (ii) sufficient phase shifting model with uniform reflection amplitude, as in \cite{nidl1}, and (iii) sufficient phase shifting model with phase-dependent amplitude, as in \cite{PracticalModelWuqq}.
The value of the reflected signal amplitudes is determined by the amplitude response and phase response of the prototypes.
The curves (1), (2), and (3) in Fig. \ref{Fig:AlgA} represent 2-bit quantized systems with states `00’, `01’, `10’, and `11’ corresponding to reflection amplitudes ${0\; \text{dB}}$, ${-6\; \text{dB}}$, ${-10\; \text{dB}}$, and ${-3\; \text{dB}}$, respectively.
The curves (4), (5), and (6) in Fig. \ref{Fig:AlgA} represent 3-bit quantized systems with corresponding reflection amplitudes of ${0 \; \text{dB}}$, ${-3\; \text{dB}}$, ${-6 \;\text{dB}}$, ${-9\; \text{dB}}$, ${-10\; \text{dB}}$, ${-7\; \text{dB}}$, ${-3\; \text{dB}}$, and ${-2\; \text{dB}}$ in the `000’, `001’, `010’, `011’, `100’, `101’, `110’, and `111’ states, respectively.
This figure shows that the proposed group-based query algorithm outperforms the schemes based on the sufficient phase shifting model.
This superiority is attributed to its ability to strike a balance between the reflected signal amplitude by individual elements and phase alignment over all the elements, in accordance with the actual amplitude and phase-shifting capability, so as to achieve the maximum signal power at the receiver.
When the number of bits $k$ increases, the performance gap between our scheme and the other two increases, especially when the phase shifting capability is relatively small.

In Fig. \ref{Fig:AlgB}, we evaluate the differences of the maximum LARP between two 2-bit quantized RIS-aided systems in which the RIS elements have different amplitudes of the reflected signal. The reflection coefficients of the RISs are determined by the proposed group-based query algorithm.
`AM[dB]’ in this figure means the value of the reflected signal amplitudes.
The `00’, `01’, `10’, and `11’ states in curve (1) correspond to the reflection amplitudes ${0\; \text{dB}}$, ${-5\; \text{dB}}$, ${-6\; \text{dB}}$, and ${-2\; \text{dB}}$, respectively. In curve (2), the states `00’, `01’, `10’, and `11’ correspond to reflection amplitudes ${0\; \text{dB}}$, ${-6\; \text{dB}}$, ${-10\; \text{dB}}$, and ${-3\; \text{dB}}$, respectively.
It can be seen from this figure that the derived theoretical expressions match well with the simulation results, which validates the proposed theorem.
Besides, from the comparison of the curves (1) and (2), we can see that when the RIS system has sufficient phase shifting capability or a comparatively large phase shifting capability, the lower reflection amplitudes lead to a lower LARP, which is expected because lower reflection signal amplitudes mean a lower reflection signal power.
Nevertheless, when the phase shift capability of the RIS system is below a specified threshold, a lower reflection amplitude of the RIS elements improves the LARP of the system, as shown in Fig. \ref{Fig:AlgB}.
This is due to the fact that when the phase shifting capability of the system is reduced to a specific level, the quantization error $\delta$ will be greater than 90° for some desired phase, indicating that the reflected signal of the RIS elements will have a negative impact on the LARP; in this case, the lower the reflection amplitudes mean the negative impact on the LARP is smaller.

 \begin{figure}[!t]
	\centering
	\subfloat[Theoretical block diagram of the prototype system.]
	{
		\includegraphics[width=2.5in]{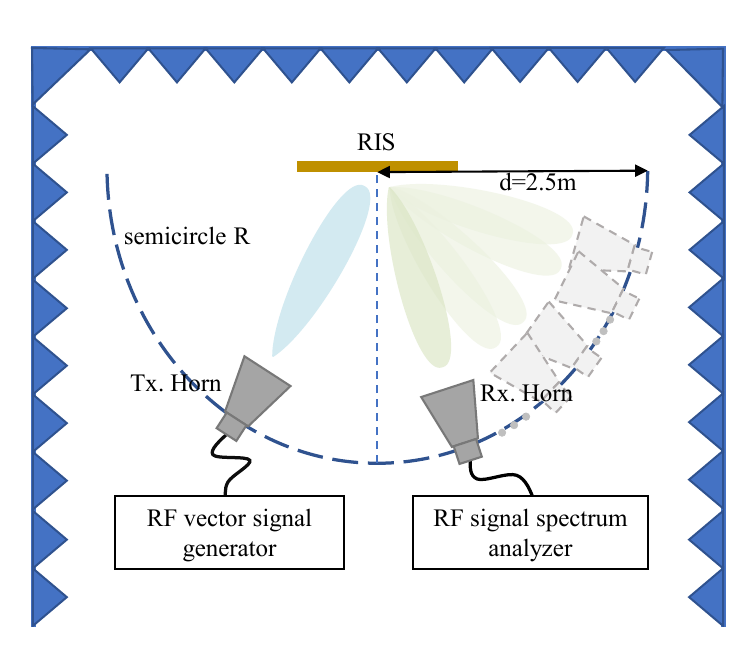}
		\label{Fig9a}
	}
	\quad
	
	\subfloat[Physical diagram of the prototype system]{
		\includegraphics[width=3.3in]{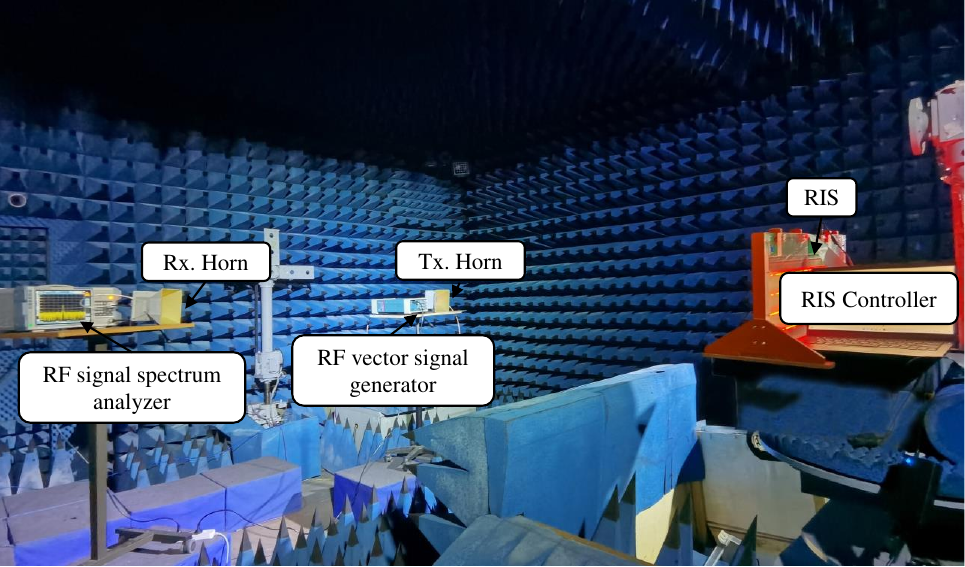}
		\label{Fig9b}
	}
	
	\caption{The measurement platform for RIS-aided wireless communications.}\label{Fig9}
\end{figure}

 \begin{figure}[!t]
    \centering
	\subfloat[The RIS working at 5.8 GHz]
	{
		\includegraphics[width=2.5in]{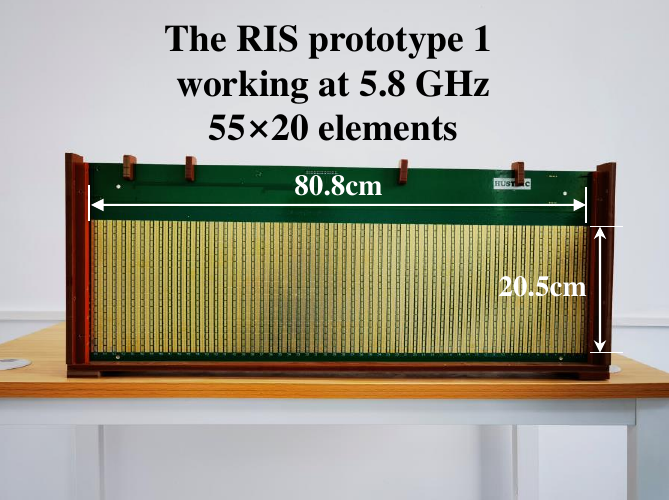}
		\label{Fig10a}
	}
	\quad
	
	\subfloat[The RIS working at 2.6 GHz]	{
		\includegraphics[width=2.5in]{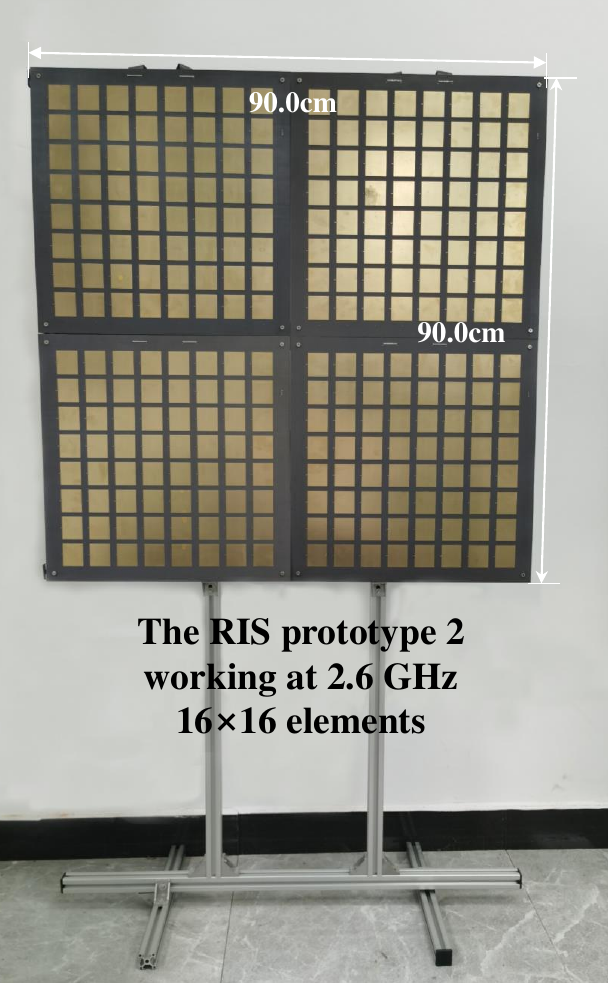}
		\label{Fig10b}
	}
	
	\caption{Photographs of the RISs utilized for the performance measurements of RIS-aided wireless communications.}\label{Fig10}
\end{figure}

\subsection{Experimental Measurements}
In this subsection, the experimental results validate the effect of realistic discrete phase shifters on the performance of RIS-aided communication systems.
We established a measurement system and employed two different RISs with non-ideal reflection coefficients. Fig. \ref{Fig9} illustrates the measurement system, which includes RIS, an RF vector signal generator (Tektronix TSG4106A), a Tx horn antenna, an RF signal spectrum analyzer (Rohde $\&$ Schwarz ZNB 8), cables, and blockages (electromagnetic absorbers). The RIS, Tx horn antenna, and Rx horn antenna are horizontally polarized and well matched in the experimental measurement. As shown in Fig. \ref{Fig9} (a), the transmitting and receiving antennas are positioned on a semicircle with RIS at the center and a radius of $ {d = 2.5\; \text{m}}$. The transmitting and receiving horn antennas are aligned with the center of the RIS. The RF vector signal generator provides the RF signal to the Tx horn antenna. The signal reflected by the RIS propagates over the distance $ d $ and is received by the Rx horn antenna and the RF signal spectrum analyzer, which gives the measured received signal power.

Fig. \ref{Fig10} shows the RISs used in different scenarios. The RIS in Fig. \ref{Fig10} (a) operates at 5.8 GHz with element sizes of ${d_h = 14.3\; \text{mm}}$, ${d_v = 10.27\; \text{mm}}$, and the number of elements ${M = 1100}$. More details of this RIS can be found in our previous work \cite{5G8Prototype}. The RIS in Fig. \ref{Fig10} (b) operates at 2.6 GHz, and has ${M = 256}$ elements with the element sizes ${d_h = 45\; \text{mm}}$ and ${d_v = 45\; \text{mm}}$. Both RISs are 1-bit regulated by varactor diodes. Since altering the bias voltage changes the impedance of the varactor diode as well as the loss induced by the dielectric substrate, metal plate, etc, the reflection phase and amplitude of the reflected signal will fluctuate unpredictably.

We measured the phase and amplitude differences of the two states of 5.8 GHz RIS at different incident angles at 3 V and 7 V bias voltages, representing the states `0’ and `1’, respectively.
As shown in Table \ref{tab:Dparameters}, due to the spatial dispersion of the RIS unit, the reflection coefficient of RIS is sensitive to incident angle, implying that an RIS system with sufficient phase shift capability at a specific incident angle may become less than satisfactory when the incident angle is altered. In other words, when the incident angle increases, the phase shifting capability of the RIS unit decreases, which will lead to the degradation of system performance.
Then, using the 5.8 GHz RIS, we conduct experiments to investigate the impact of incidence angle change on the received power and compare the measured results to those calculated according to proposition \ref{PropNIdl}.
In this scenario, for evaluating system performance degradation, the system performance at 10° incidence is served as the baseline.
We move the transmitting antenna to various angles and select 12 random places on the circular arc $R$ illustrated in Fig. \ref{Fig9} (a) to measure the power after RIS beamforming at each of these points and average the results.
As shown in Table \ref{tab:Dparameters}, the system with sufficient phase shifting capability at 10° incidence has a diminishing trend in beamforming capability as the incident angle increases.
As shown in Fig. \ref{fig_11}, the system performance decreases as the angle of incidence increases.
Besides, the measured curve follows the same trend as the theoretical curve, with the biggest difference being only about 0.3 dB. This discrepancy may be due to environmental factors.

\begin{figure}[!t]
\centering
\includegraphics[width=3.3in]{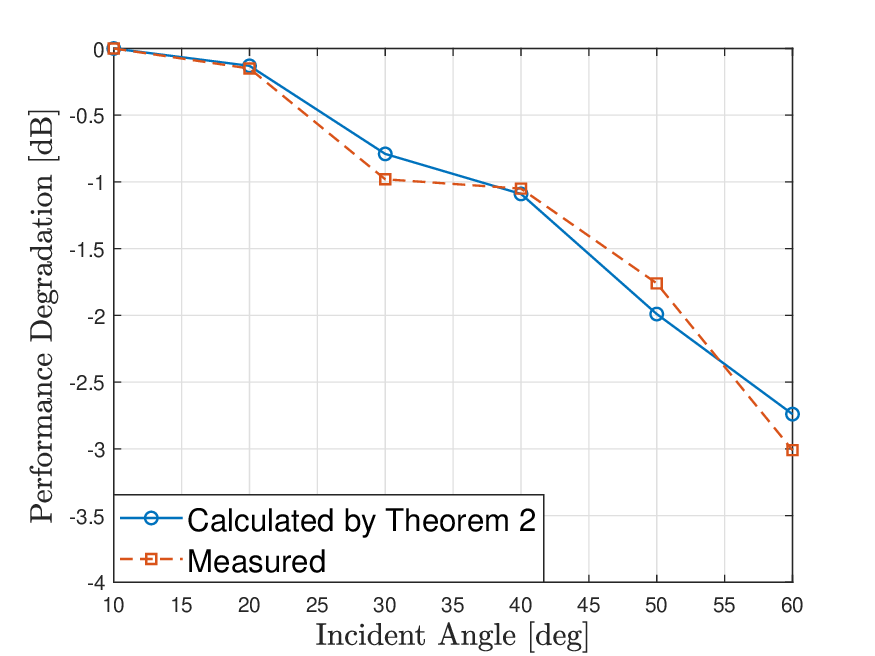}
\DeclareGraphicsExtensions.
\caption{Performance degradation versus incident angle.}
\label{fig_11}
\end{figure}

 \begin{table}[!t] 
	\renewcommand{\arraystretch}{1.3}
	\caption{Measured Reflection Coefficients}
	\label{tab:Dparameters}
	\centering
	\begin{tabular}{|ccccccc|}
		\hline\hline
		\textbf{Incident angle (degree)}      & 10   & 20   & 30   & 40     & 50   & 60     \\ \hline
		\textbf{Phase difference (degree)}    & 180   & 160   & 132   & 117     & 107   & 76     \\ \hline
		\textbf{Amplitude difference (dB)}    & 2     & 0.7     & 0.1     & 0.3      & 2.3    & 1.5      \\ \hline
		\hline
		
	\end{tabular}
\end{table}
 
\begin{figure}[!t]
\centering
\includegraphics[width=3in]{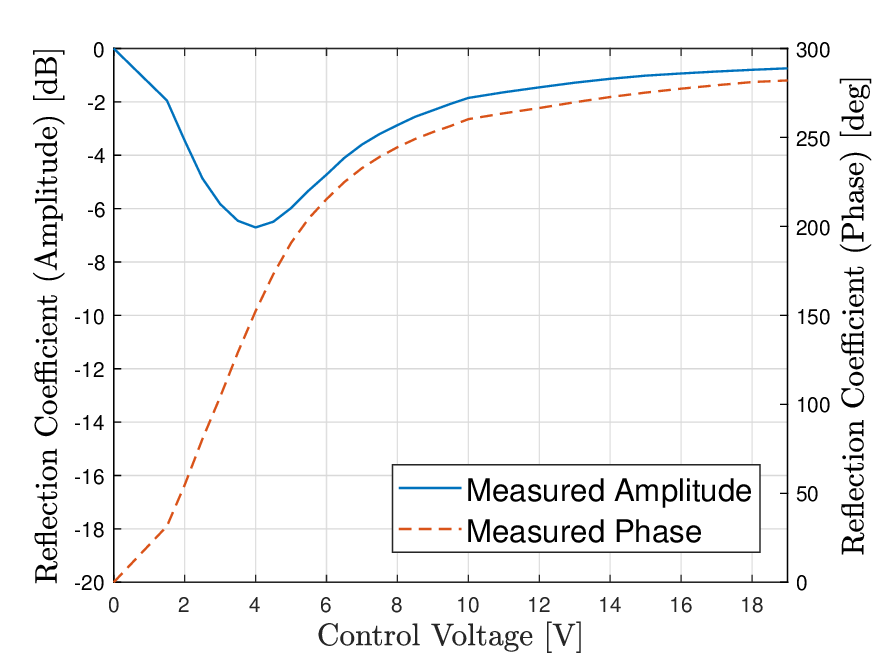}
\DeclareGraphicsExtensions.
\caption{The relationship between the control voltage and the coefficient of the elements for the 2.6 GHz RIS.}
\label{fig_12}
\end{figure}
 
\begin{figure}[!t]
\centering
\includegraphics[width=3.3in]{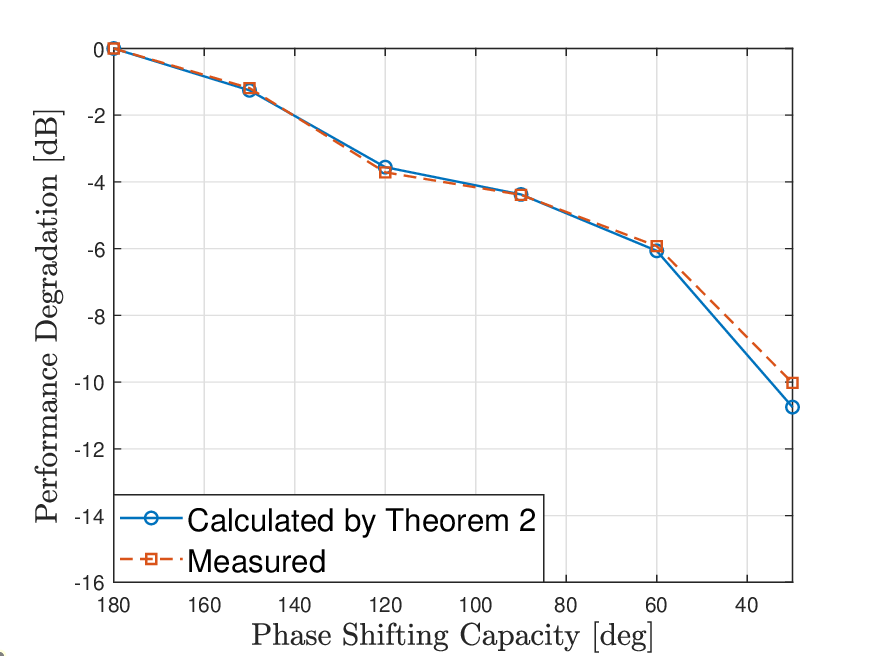}
\DeclareGraphicsExtensions.
\caption{Performance degradation versus phase shifting capacity.}
\label{fig_13}
\end{figure}

Furthermore, based on the 2.6 GHz RIS prototype, we simulate that the RIS system lacks sufficient phase shifting capability due to the RIS element design by varying the bias voltage corresponding to the states `0’ and `1’, evaluate the performance of the RIS system with different phase shifting capabilities, and compare it to theoretical results.
We fixed the incident angle at 10° and measured the relationship between the 2.6 GHz RIS control voltage and the phase shift and amplitude of each RIS element; the results are presented in Fig. \ref{fig_12}. The bias voltage is then manually adjusted to vary the phase difference between the state `0’ and the state `1’ of the RIS.
Six sets of bias voltages based on the measured results in Fig. \ref{fig_12} are chosen to make the phase difference 180°, 150°, 120°, 90°, 60°, and 30°, respectively.
In this scenario, the received power of RIS at a set of bias voltages that produced a 180° phase difference is served as a baseline for performance comparison.
For each pair of bias voltages, 12 positions are chosen at random on the circular arc R shown in Fig. \ref{Fig9} (a). The received power after RIS beamforming is measured at these positions, and the results are averaged to obtain the measurement results as shown in Fig. \ref{fig_13}.
We observe that the beamforming capability of the RIS diminishes as the phase shifting capability of the system decreases. Hence, sufficient phase shift capability must be ensured at the RIS element design stage.
The measured curve has the same trend as the theoretical curve, and the biggest difference is only about 0.3 dB. This discrepancy in results, like in the previous experiment, could be attributed to environmental influences.

 \newcounter{TempEqCntEq25} 
 \begin{figure*}[hb] 
 	\hrulefill 
 	\begin{align}\label{eq25}
 		\begin{split}
 			{\left\|{{\mathbf{f} ^H}\mathbf{\Phi h}}\right\|}^2 
 			&= {\left\| {{{\left( \sqrt {\frac{{K_2}}{{K_2} + 1}}\mathbf{\bar f}+\sqrt {\frac{1}{{K_2} + 1}}\mathbf{\tilde f } \right)}^H}\mathbf{\Phi}  \left( \sqrt {\frac{{K_1}}{{K_1} + 1}}\mathbf{\bar h}+\sqrt {\frac{1}{{K_1} + 1}}\mathbf{\tilde h } \right)} \right\|^2}\\
 			&= \frac{1}{\left({K_1 + 1}\right)\left({K_2+1}\right)}{\left\|{\underbrace{\sqrt{{K_1}{K_2}}{\mathbf{\bar f} ^H}\mathbf{\Phi \bar h} }_{{x_1}} + \underbrace {\sqrt {{K_1}} {{\mathbf{\tilde f} }^H}\mathbf{\Phi \bar h} }_{{x_2}} + \underbrace {\sqrt {{K_2}} {{\mathbf{\bar f}}^H}\mathbf{\Phi \tilde h} }_{x_3} + \underbrace{{{\mathbf{\tilde f}}^H}\mathbf{\Phi \tilde h} }_{x_4}} \right\|^2}.
 		\end{split}
 	\end{align}	
 \end{figure*}
 
\section{Conclusion}\label{Sec:Conclusion}
In this paper, we proposed a realistic reflection coefficient model for RIS-aided wireless communication systems, which takes into account the discreteness of the phase shift, the attenuation of the reflective signal and the limited phase shifting capability.
The maximum received power of the user based on this model was derived.
We then proposed a group-based query algorithm to maximize the received power for the RIS-aided system with the realistic reflection coefficient model.
We analyzed the asymptotic performance of the proposed algorithm and derived the closed-form expression for the maximum long-term average received power.
Finally, by conducting both simulations and corresponding experiments with the fabricated RIS prototype systems, we verified the proposed theoretical results.
The simulated results and measurement results all match quite well with the analytical results.

%

%
\appendices
\section{Proof of Proposition 1}\label{Append: Prop1}
By applying (\ref{eq2}), (\ref{eq3}) in ${\left\| {{\mathbf{f} ^H}\mathbf{\Phi h}}\right\|}^2$, Eq. (\ref{eq25}) at the bottom of the next page holds.

Therefore, $\mathbb{E}\left\{{{{\left\| {{\mathbf{f} ^H}\mathbf{\Phi h}}\right\|}^2}}\right\}$ is given by
\begin{align}\label{eq26}
	\begin{split}
	\mathbb{E}&\left\{{{{\left\| {{\mathbf{f} ^H}\mathbf{\Phi h}}\right\|}^2}}\right\} = \frac{1}{{\left( {{K_1} + 1} \right)\left( {{K_2} + 1} \right) }} \\
	&\times \left( {\sum\limits_{i = 1}^4 {\mathbb{E}\left\{ {{{\left\| {{x_i}} \right\|}^2}} \right\}}  + \sum\limits_{i = 1,j = 1,i \ne j}^4 {\mathbb{E}\left\{ {{x_i}^H{x_j}} \right\}} } \right).
	\end{split}
\end{align}	

Since the NLoS components are independent with each other, and have zero means, any correlation between channel matrices is zero. We observe that
\begin{equation}\label{eq27}
	\mathbb{E}\left\{ {{x_i}^H{x_j}} \right\} = 0,i,j = 1,2,3,4,i \ne j.
\end{equation}

For the LoS channel, by applying (\ref{eq4}) and (\ref{eq5}) in $\mathbb{E}\left\{ {{{\left\| {{x_1}} \right\|}^2}} \right\}$, we may derive
\begin{align}\label{eq28}
	\begin{split}
	\mathbb{E}\left\{ {{{\left\| {{x_1}} \right\|}^2}} \right\} &= {K_1}{K_2} {\left( {\mathbf{\bar f} ^H}\mathbf{\Phi \bar h} \right)^H}\left( {\mathbf{\bar f} ^H}\mathbf{\Phi \bar h}\right)\\
	&= {K_1}{K_2} {\sqrt {D_1^{-\alpha }d_1^{-\alpha }}} {G_a} \\
	&\times \left( {\sum\limits_{m,{m'}}}{{{A_m}{A_{m'}}{e^{ - j[{\phi _m} - {\phi _{m'}} + {\theta _m} - {\theta _{m'}}]}}} }\right).
	\end{split}
\end{align}	

According to the random matrix theory, it is easy to obtain
\begin{align}\label{eq29}
	\begin{split}
	&\mathbb{E}\left\{ {{\mathbf{\Phi} ^H}\mathbf{\Phi} } \right\} = diag\left( A_1^2,\cdots,A_m^2,\cdots,A_M^2\right),\\
	&\mathbb{E}\left\{ {\mathbf{\tilde f{{\tilde f}}^H}} \right\} = M,\\
	&\mathbb{E}\left\{ {{{\mathbf{\tilde h}}^H}\mathbf{\tilde h}} \right\} = 1.
	\end{split}
\end{align}	

Thus, we may derive
\begin{align}\label{eq30}
	\begin{split}
	\mathbb{E}\left\{ {{{\left\| {{x_2}} \right\|}^2}} \right\} &={K_1} \mathbb{E}\left\{ {{\mathbf{\bar h}^H}\mathbf{\Phi}^H \mathbf{\tilde f}{\mathbf{\tilde f}^H}\mathbf{\Phi \bar h}} \right\}\\
	&= {K_1} {\sqrt {{G_a}D_1^{-\alpha }}}{L( d_1 )}\sum\limits_m {A_{{m}}^2}. 
	\end{split}
\end{align}	

Similarly,
\begin{align}\label{eq31}
	\begin{split}
	\mathbb{E}\left\{ {{{\left\| {{x_3}} \right\|}^2}} \right\} &= {K_2} \mathbb{E}\left\{ {{\mathbf{\tilde h}^H}{\mathbf{\Phi} ^H}\mathbf{\bar f}{{\mathbf{\bar f}}^H}\mathbf{\Phi \tilde h}} \right\}\\
	&= {K_2} {\sqrt {{G_a}d_1^{-\alpha }}}{L( D_1 )}\sum\limits_m {A_{{m}}^2}, 
	\end{split}
\end{align}	
\begin{align}\label{eq32}
	\begin{split}
	\mathbb{E}\left\{ {{{\left\| {{x_4}} \right\|}^2}} \right\} &= \mathbb{E}\left\{ {{\mathbf{\tilde h}^H}{\mathbf{\Phi} ^H}\mathbf{\tilde f}{\mathbf{\tilde f}^H}\mathbf{\Phi \tilde h}} \right\}\\
	&={L( D_1 )}{L( d_1 )} \sum\limits_m {A_{{m}}^2}.  
	\end{split}
\end{align}	

Then, by applying (\ref{eq26}), (\ref{eq28})-(\ref{eq32}) to (\ref{eq10}), we obtain:
\begin{align}\label{eq33}
	\begin{split}
		\Gamma = \kappa _\text{NLoS}\sum\limits_m {A_{{m}}^2} + \kappa _\text{LoS} \sum\limits_{m,m'} {{A_m}{A_{m'}}{e^{ - j[{\phi _m} - {\phi _{m'}} + {\theta _m} - {\theta _{m'}}]}}}. 
	\end{split}
\end{align}	

This ends the proof.\hfill {$\square$}

\newcounter{TempEqCntEq37} 
\begin{figure*}[hb] 
	\hrulefill  
	\begin{align}\label{eq37}
		\mathbb{E}\left[ {\varphi({\delta _m},{\delta _{m'}})} \right] =
		\left\{
		\begin{aligned}
			&\int_{-\frac{\pi}{2^{{k}}}}^{\frac{\pi}{2^{{k}}}}{\int_{-\frac{\pi }{2^{{k}}}}^{\frac{\pi }{2^{{k}}}} {\frac{2^{{{2k - 2}}}}{\pi ^2}\varphi({\delta _m},{\delta _{m'}})} } d{\delta _m}d{\delta _{m'}},\quad \omega  \ge \frac{{2^{{k}}} - 1}{2^{{k - 1}}}\pi ;\\
			&\int_{ - a }^{ - b } {\int_{-a }^{ - b } {{P_2}^2\varphi({\delta _m},{\delta _{m'}})} } d{\delta _m}d{\delta _{m'}} + \int_{ - a }^{ - b } {\int_{ - b}^b  {{P_2}{P_1}\varphi({\delta _m},{\delta _{m'}})} } d{\delta _m}d{\delta _{m'}} +\\
			&\int_{-a}^{-b}{\int_b^a{{P_2}^2 \varphi({\delta _m},{\delta _{m'}})}} d{\delta _m}d{\delta _{m'}} + \int_{- b}^b {\int_{-a}^{ -b } {{P_1}{P_2} \varphi({\delta _m},{\delta _{m'}})} } d{\delta _m}d{\delta _{m'}} + \\
			&\int_{ -b }^b  {\int_{ -b }^b  {{P_1}^2\varphi ({\delta _m},{\delta _{m'}})} } d{\delta _m}d{\delta _{m'}}+\int_{ - b }^b  {\int_b ^a  {{P_1}{P_2}\varphi({\delta _m},{\delta _{m'}})} } d{\delta _m}d{\delta _{m'}} + \\
			&\int_b ^a  {\int_{ - a }^{ - b } {{P_2}^2 \varphi({\delta _m},{\delta _{m'}})} } d{\delta _m}d{\delta _{m'}}+\int_b ^a  {\int_{ - b }^b  {{P_2} {P_1}\varphi({\delta _m},{\delta _{m'}})} } d{\delta _m}d{\delta _{m'}} + \\
			&\int_b ^a  {\int_b ^a  {{P_2}^2\varphi({\delta _m},{\delta _{m'}})} } d{\delta _m}d{\delta _{m'}},\quad\omega  < \frac{{2^{{k}}} - 1}{2^{{{k - 1}}}}\pi,
		\end{aligned}
		\right.
	\end{align}
	where $a  = \pi  - \frac{\omega }{2}$, $b  = \frac{\omega }{{2\left( {{2^{{k}}} - 1} \right)}}$, ${P_1} = \frac{{{2^{{k}}}}}{{2\pi }}$, ${P_2} = \frac{1}{{2\pi }}$.
\end{figure*}

\section{Proof of Theorem 1}\label{Append: Theor1}
We first define a function $\varphi({\delta _m},{\delta _{m'}})$ as
\begin{equation}\label{eqIDL}
	\varphi({\delta _m},{\delta _{m'}}) = \sin {\delta _m}\sin {\delta _{m'}} + \cos {\delta _m}\cos {\delta _{m'}}.
\end{equation}
Then according to the expression of LARP showed in (\ref{eq11}), we obtain the expectation of the maximum LARP $\mathbb{E}\left( {\hat \Gamma_{\rm{max}} } \right)$:
\begin{align}\label{eq34}
	\begin{split}
		\mathbb{E}\left( {\hat \Gamma_{\rm{max}} } \right) = \kappa _\text{NLoS}M+  \kappa _\text{LoS}{M^2}\mathbb{E}\left[ {\varphi({\delta _m},{\delta _{m'}})} \right]. 
	\end{split}
\end{align}	

Since $K_1$, $K_2$, $\eta _\text{NLoS1}$, $\eta _\text{NLoS2}$, $\eta _\text{NLoS3}$, $\eta _\text{LoS}$ and $M$ are constant, we will focus on $\mathbb{E}\left[ {\varphi({\delta _m},{\delta _{m'}})} \right]$ in the following.

Because the expected phase shift $\theta _m^* $ is uniformly distributed in $ [0,2\pi )$ and the discrete phase shift closest to the expected one will be chosen to achieve the maximum LARP, depending on phase shifting capability $\omega$, the quantization error $\delta$ is uniformly distributed on $\left[ { - \frac{\pi }{{{2^{{k}}}}},\frac{\pi }{{{2^{{k}}}}}} \right)$, or uniformly distributed over each of three contiguous subintervals that are $\left[ {\frac{\omega }{2} - \pi ,- \frac{\omega }{{2\left( {{2^{{k}}} - 1} \right)}}} \right)$, $\left[ { - \frac{\omega }{{2\left( {{2^{{k}}} - 1} \right)}} , \frac{\omega }{{2\left( {{2^{{k}}} - 1} \right)}}} \right)$ and $\left[ {\frac{\omega }{{2\left( {{2^{{k}}} - 1} \right)}} , \pi  - \frac{\omega }{2}} \right)$. When $\omega  < \frac{{{2^{{k}}} - 1}}{{{2^{{{k - 1}}}}}}\pi $, the probability density function (PDF) of $\delta$ is obtained as 
\begin{align}\label{eq36}
	&f_{\delta }\left( \delta  \right) \nonumber\\
	&=\left\{
	\begin{aligned}
		&\frac{{{2^{{k}}}}}{{2\pi }},\quad \delta  \in \left[ { - \frac{\omega }{{2\left( {{2^{{k}}} - 1} \right)}},\frac{\omega }{{2\left( {{2^{{k}}} - 1} \right)}}} \right);\\
		&\frac{1}{{2\pi }},\quad \delta  \in \left[ {\frac{\omega }{2} - \pi, - \frac{\omega }{{2\left( {{2^{{k}}} - 1} \right)}}} \right) \cup \left[ {\frac{\omega }{{2\left( {{2^{{k}}} - 1} \right)}}, \pi  - \frac{\omega }{2}} \right);\\
		&0,\quad {\rm{others}}.
	\end{aligned}
	\right.
\end{align}
Otherwise, when $\omega  \ge \frac{{{2^{{k}}} - 1}}{{{2^{{{k - 1}}}}}}\pi $, the PDF of $\delta$ is obtained as
\begin{equation}\label{eq35}
	f_{\delta }\left( \delta  \right) =\left\{
	\begin{aligned}
		&\frac{{{2^{{k}}}}}{{2\pi }}, \quad \delta  \in \left[ { - \frac{\pi }{{{2^{{k}}}}},\frac{\pi }{{{2^{{k}}}}}} \right);\\
		& 0, \quad {\rm{others}}.
	\end{aligned}
	\right.
\end{equation}

 $\mathbb{E}\left[ {\varphi({\delta _m},{\delta _{m'}})} \right]$ is expressed as (\ref{eq37}) in the RIS-aided system with phase shifting capability $\omega$ in terms of PDF of $\delta$, which is shown at the bottom of the next page.

Applying (\ref{eqIDL}) to (\ref{eq37}) and through some basic algebraic manipulations, we derive that
\begin{align}\label{eq38}
	&\mathbb{E}\left[ {\varphi({\delta _m},{\delta _{m'}})} \right]\nonumber\\
	&=\left\{
	\begin{aligned}
		&\frac{2^{{2k}}}{\pi ^2}{\sin ^2}\frac{\pi }{2^{{k}}},\quad \omega\ge \frac{2^{{k}} - 1}{2^{{k - 1}}}\pi ;\\
		& 4\left[{P_1}\sin b + {P_2}(\sin a - \sin b) \right]^2,\quad\omega  < \frac{2^{{k}} - 1}{2^{{k - 1}}}\pi . 
	\end{aligned}
	\right.
\end{align}

By substituting (\ref{eq38}) into (\ref{eq34}), $\mathbb{E}\left( {\hat \Gamma_{\rm{max}} } \right)$ is obtained as (\ref{eq17}). This ends the proof.\hfill {$\square$}

\section{Proof of Theorem 2}\label{Append: Theor2}
 According to the expression of LARP shown in (\ref{eq11}), the maximum LARP expectation $\mathbb{E}(\hat \Gamma_{\rm{max}})$ in this scenario is given by
\begin{align}\label{eq39}
	\begin{split}
		\mathbb{E}(\hat \Gamma_{\rm{max}})
		=&\kappa _\text{NLoS}M\mathbb{E}\left[ A_m^2\right]+ \kappa _\text{LoS}{M^2}\mathbb{E}\left[ {{\varrho }({\delta _m},{\delta _{m'}})} \right],
	\end{split}
\end{align}
where 
\begin{equation}\label{eq40}
	{\varrho }({\delta _m},{\delta _{m'}}) = {A_m}{A_{m'}}\left( {\sin {\delta _m}\sin {\delta _{m'}} + \cos {\delta _m}\cos {\delta _{m'}}} \right).
\end{equation}

As $K_1$, $K_2$, $\eta _\text{NLoS1}$, $\eta _\text{NLoS2}$, $\eta _\text{NLoS3}$, $\eta _\text{LoS}$ and M are constant, the key to derive $\mathbb{E}(\hat \Gamma_{\rm{max}})$ is to derive $\mathbb{E}\left[ A_m^2\right]$ and $\mathbb{E}\left[ {{\varrho }({\delta _m},{\delta _{m'}})} \right]$.

Since the optimal phase shift ${\theta}^*$ is uniformly distributed in $[0,2\pi )$, the probability density of ${\theta}^*$ in $[0,2\pi )$ is always $1/2\pi$. 
Therefore, the probability of using the $i-$th quantified reflection coefficient is obtained as
\begin{equation}\label{eqPI}
	{P_i} = \frac{\mu _i}{2\pi},
\end{equation}
where $\mu _i$ denotes the length of ${c_i}$ which is the optimal phase shift range corresponding to each quantized reflection coefficient $\Phi_i$ and is obtained from the group-based query algorithm. 

Thus, $\mathbb{E}\left[ A_m^2\right]$ is obtained as
\begin{equation}\label{eAm}
	\mathbb{E}\left[ A_m^2\right]=\sum\limits_{i} {	{P_i}A_{\theta _i}^2 } = \sum\limits_{i} {\frac{\mu_i}{2\pi}A_{\theta _i}^2 },
\end{equation}

For each ${c_i}$, the corresponding quantization error $d_i$ is denoted by
\begin{equation}\label{eq41}
	{d_i} = {c_i} - {\theta _i}.
\end{equation}

Since the probability density of ${\theta}^*$ in $[0,2\pi )$ is always $1/2\pi$, the joint probability density of ${\delta _m}$ and ${\delta _{m'}}$ is $1/4\pi^2$ when $(\delta_m \in d_i) \cup (\delta_{m'} \in d_{i'}) , i,i'=1,\cdots,2^k$.

Thus, $\mathbb{E}\left[ {{\varrho }({\delta _m},{\delta _{m'}})} \right]$ is obtained as
\begin{align}\label{eq42}
	&\mathbb{E}\left[ {{\varrho }\left( {{\delta _m},{\delta _{m'}}} \right)} \right]\nonumber\\
	&= \sum\limits_{i,i'} \int_{{\delta _m} \in {d_i}} {\int_{{\delta _{m'}} \in {d_{i'}}} {{\frac{1}{{4{\pi ^2}}}{\varrho }\left( {{\delta _m},{\delta _{m'}}} \right)d{\delta _m}d{\delta _{m'}}} } },
\end{align}

By substituting (\ref{eAm}) and (\ref{eq42}) into (\ref{eq39}), the LARP expectation $\mathbb{E}\left( {{{\hat \Gamma }_{\max }}} \right)$ is obtained as (\ref{eq21}). This ends the proof.\hfill {$\square$}

\section{Proof of Corollary 1}\label{Append: Cor1}

We first define $\omega ' = \min \left( {\omega ,\pi } \right)$. When the RIS is 1-bit coded and the channels are pure LoS paths, the optimization problem (P3) of the reflection coefficient of the $m-$th element is simplified to:
\begin{equation}\label{eq43}
	{\hat{\phi}^*_m} = \mathop {\max } \left[ {{A_{\theta_1}}\cos ( - {\theta ^*_m}),{A_{\theta_2}}\cos (\omega ' - {\theta ^*_m})} \right],
\end{equation}
which can be rewritten as
\begin{equation}\label{eq44}
	{\hat{\phi}^*_m} =\left\{
	\begin{aligned}
		&{\Phi _1},\left[ {{{A_{\theta_1}}}\cos ( - {\theta ^*_m}) - {A_{\theta_2}}\cos (\omega ' - {\theta ^*_m})} \right] > 0;\\
		&{\Phi _2},\left[ {{A_{\theta_1}}\cos ( - {\theta ^*_m}) - {A_{\theta_2}}\cos (\omega ' - {\theta ^*_m})} \right] < 0.
	\end{aligned}
	\right.
\end{equation}

To solve the optimal phase shift range $c_i$, we define a function as
\begin{equation}\label{eq_slvci}
	{\zeta(\theta)} = {A_{\theta_1}}\cos ( - {\theta}) - {A_{\theta_2}}\cos (\omega ' - {\theta}).
\end{equation}
Eq. (\ref{eq_slvci}) can be converted to
\begin{align}\label{eq45}
	{\zeta(\theta)}=&-\sqrt {{A_{\theta_2}}^2{{\sin }^2}{\omega '} + {{\left( {{A_{\theta_2}}\cos {\omega '} - {A_{\theta_1}}} \right)}^2}}\nonumber\\
	&\times\sin \left( {{\theta} + \arctan \frac{{{A_{\theta_2}}\cos {\omega '} - {A_{\theta_1}}}}{{A_{\theta_2}}{\sin {\omega '}}}} \right).
\end{align}

Since (\ref{eq45}) is obviously in sine form, we rewrite (\ref{eq44}) as 
\begin{equation}\label{eq46}
	{\hat{\phi}^*_m} =\left\{
	\begin{aligned}
		&{A_{\theta_1}}{e^{ - j{\omega '} \cdot 0}},{\theta ^*_m} \in \left[ {0,{\psi _1}} \right) \cup \left[ {{\psi _2},2\pi } \right);\\
		&{A_{\theta_2}}{e^{ - j{\omega '} \cdot 1}},{\theta ^*_m} \in \left[ {{\psi _1},{\psi _2}} \right).
	\end{aligned}
	\right.
\end{equation}
where ${\psi _1}$ and ${\psi _2}$ are denoted by
\begin{align}\label{eq:psi}
    &{\psi _1} =  -  \arctan \frac{{{A_{\theta_2}}\cos {\omega '} - {A_{\theta_1}}}}{{A_{\theta_2}}{\sin {\omega '}}} \in \left[ 0, \frac{\pi}{2} \right)\cup \left( \frac{\pi}{2}, \pi \right),\\
    &{\psi _2} = \pi  -  \arctan \frac{{{A_{\theta_2}}\cos {\omega '} - {A_{\theta_1}}}}{{A_{\theta_2}}{\sin {\omega '}}} \in \left[ \pi, \frac{3\pi}{2} \right)\cup \left( \frac{3\pi}{2}, 2\pi \right).
\end{align}

Applying $c_1 \in \left[ {0,{\psi _1}} \right) \cup \left[ {{\psi _2},2\pi } \right)$ and $c_2 \in  \left[ {{\psi _1},{\psi _2}} \right)$ in (\ref{eq42}), we derive that
\begin{equation}\label{eq47}
	\mathbb{E}\left[ {{\varrho}\left( {{\delta _m},{\delta _{m'}}} \right)} \right] = \frac{{A_{{\theta _1}}^2 + A_{{\theta _2}}^2 - 2{A_{{\theta _1}}}{A_{{\theta _2}}}\cos \omega '}}{{\pi ^2}}.
\end{equation}

By substituting (\ref{eq47}) into (\ref{eq39}), $\mathbb{E}\left( \hat{\Gamma}_{max}\right)$ is obtained as (\ref{eq22}). This ends the proof.\hfill {$\square$}



\ifCLASSOPTIONcaptionsoff
  \newpage
\fi



%
%
%
%
\bibliographystyle{IEEEtran}
\bibliography{references}

\begin{IEEEbiography}[{\includegraphics[width=1in,height=1.25in,clip,keepaspectratio]{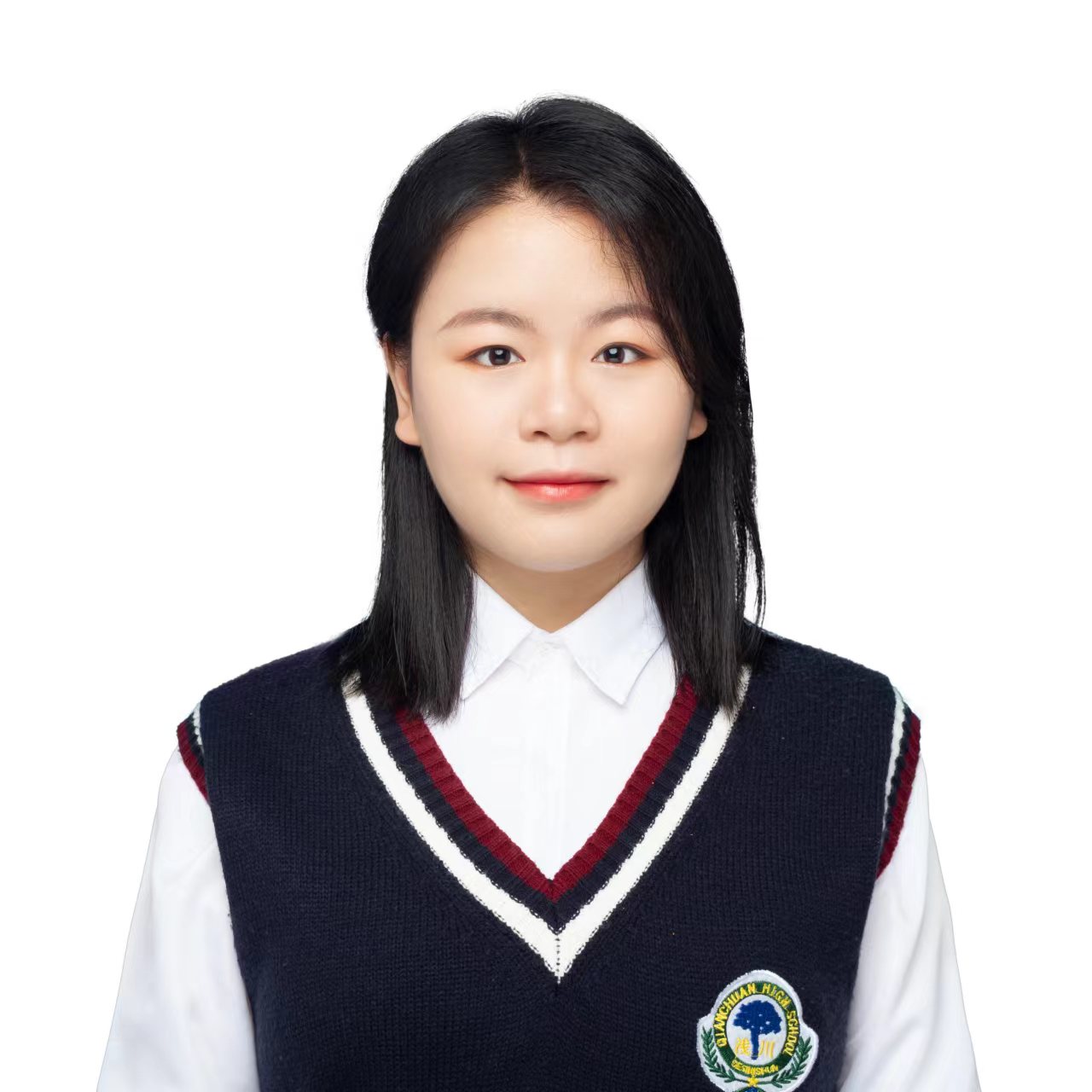}}]{Lin Cao}
	received her B.Sc. degree in Electronic Science and Technology from the School of Optical and Electronic Information, Huazhong University of Science and Technology, Wuhan, China, in 2020. She is currently a Ph.D. student in the School of Electronic Information and Communications, Huazhong University of Science and Technology, Wuhan, China. Her research interests include analysis and applications of reconfigurable intelligent surface.
\end{IEEEbiography}

\begin{IEEEbiography}[{\includegraphics[width=1in,height=1.25in,clip,keepaspectratio]{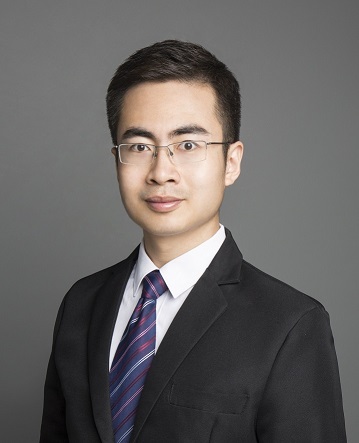}}]{Haifan Yin}
	(Senior Member, IEEE) received the B.Sc. degree in electrical and electronic engineering and the M.Sc. degree in electronics and information engineering from the Huazhong University of Science and Technology, Wuhan, China, in 2009 and 2012, respectively, and the Ph.D. degree from T\'el\'ecom ParisTech ParisTech in 2015. From 2009 to 2011, he was a Research and Development Engineer with the Wuhan National Laboratory for Optoelectronics, Wuhan, working on the implementation of TD-LTE systems. From 2016 to 2017, he was a DSP Engineer at Sequans Communications (IoT chipmaker), Paris, France. From 2017 to 2019, he was a Senior Research Engineer working on 5G standardization at Shanghai Huawei Technologies Company Ltd., where he has made substantial contributions to 5G standards, particularly the 5G codebooks. Since May 2019, he has been a Full Professor with the School of Electronic Information and Communications, Huazhong University of Science and Technology. His current research interests include 5G and 6G networks, signal processing, machine learning, and massive MIMO systems. He was the National Champion of 2021 High Potential Innovation Prize awarded by the Chinese Academy of Engineering, a recipient of the China Youth May Fourth Medal (the top honor for young Chinese), and a recipient of the 2015 Chinese Government Award for Outstanding Selffinanced Students Abroad.
\end{IEEEbiography}

\begin{IEEEbiography}[{\includegraphics[width=1in,height=1.25in,clip,keepaspectratio]{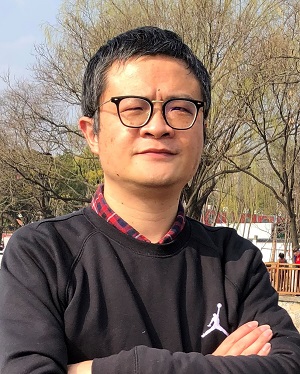}}]{Li Tan}
	was born in Hubei, China, in June 1976. He received the B.Sc. degree in Telecommunications Engineering, the M.Sc. degree in Electrical Circuit and System, and the D.Sc. degree in Information and Communication Engineering from Huazhong University of Science and Technology, Wuhan, China, in 1999, 2002, and 2009, respectively.
	Since July 1999, he has been a Lecturer with Huazhong University of Science and Technology. His current research interests include wireless communication, radio resource management, energy efficiency, reconfigurable intelligent surface, stochastic computing, and probabilistic CMOS.
\end{IEEEbiography}

\begin{IEEEbiography}[{\includegraphics[width=1in,height=1.25in,clip,keepaspectratio]{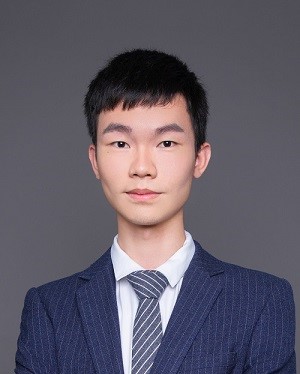}}]{Xilong Pei}
	received the B.Sc. (with highest honors) in Automation from Huazhong University of Science and Technology, China, in 2021. He is currently a Ph.D. student in the School of Electronic Information and Communications, Huazhong University of Science and Technology, China. His research interests include wireless communications, signal processing, and reconfigurable intelligent surface.
\end{IEEEbiography}

%








\end{document}